\documentclass[journal,12pt,onecolumn,draftcls]{IEEEtran}
\usepackage{graphics}
\usepackage{multicol}
\usepackage{lipsum}
\usepackage{color}
\usepackage{mathrsfs}
\usepackage{mathtools}
\usepackage{amsbsy}
\usepackage[colorlinks=true,bookmarks=false,citecolor=blue,urlcolor=blue]{hyperref}
\usepackage[dvipsnames]{xcolor}  
\usepackage{amsthm}
\usepackage{amssymb}
\usepackage{mathrsfs}
\usepackage{setspace}
\usepackage{float}
\usepackage{graphicx}
\usepackage{pstool}
\usepackage{cite}
\usepackage{amsmath,amssymb,amsfonts}
\usepackage[mathscr]{eucal} 
\usepackage{graphicx}
\usepackage{textcomp}
\usepackage{amsthm}
\usepackage{lettrine}
\usepackage[normalem]{ulem}
\usepackage{latexsym}
\usepackage{algpseudocode}
\usepackage{algorithm,algpseudocode}
\usepackage{algorithmicx}
\usepackage{multirow}
\usepackage{amsmath}
\usepackage{mathrsfs}
\usepackage{amsmath,cite,amsfonts,amssymb,color}
\usepackage{tikz}
\usetikzlibrary {patterns,patterns.meta}
\usetikzlibrary{arrows.meta}
\usetikzlibrary{decorations.pathreplacing,calligraphy}

\ifCLASSOPTIONcompsoc
\usepackage[caption=false,font=normalsize,labelfon
t=sf,textfont=sf]{subfig}
\else
\usepackage[caption=false,font=footnotesize]{subfig}
\fi
\allowdisplaybreaks

\newtheorem{theorem}{Theorem}
\newtheorem{lemma}{Lemma}

\newcommand\coolover[2]{\mathrlap{\smash{\overbrace{\phantom{%
    \begin{matrix} #2 \end{matrix}}}^{\mbox{$#1$}}}}#2}

\newcommand\coolrightbrace[2]{%
\left.\vphantom{\begin{matrix} #1 \end{matrix}}\right\}#2}

\newcommand{\Ao}{\mathcal{A}_{\mathrm{o}}} 
\newcommand{\Po}{\mathbf{P}_{\mathrm{o}}} 
\newcommand{\Zo}{\mathbf{Z}_\mathrm{o}} 

\graphicspath{{./Figures/}{./Figures/Body/}{./Figures/Bios/}{./Figures/Simulations/}}
\definecolor{azure}{rgb}{0.0, 0.5, 1.0}
\definecolor{violet}{rgb}{0.58, 0.0, 0.83}

\begin{document}
\title{In-sector Compressive Beam Alignment \\ for  MmWave and THz Radios}
\author{Hamed Masoumi, Michel Verhaegen and Nitin Jonathan Myers
\thanks{Hamed Masoumi (h.masoumi@tudelft.nl), Michel Verhaegen (m.verhaegen@tudelft.nl) and Nitin Jonathan Myers (n.j.myers@tudelft.nl) are with Delft Center for Systems and Control,
Delft University of Technology, 2628 Delft, The Netherlands. The material in this paper appeared in part in proceedings of the IEEE Signal Processing Advances in Wireless Communications (SPAWC) 2022 \cite{masoumi2022structured} and the IEEE Statistical Signal Processing (SSP) 2023 workshop \cite{masoumi2023analysis}.}}

\markboth{For submission to IEEE Trans. Wireless Commun. v1}%
{Shell \MakeLowercase{\textit{et al.}}: Bare Demo of IEEEtran.cls for IEEE Journals}

\maketitle

\begin{abstract}
Beam alignment is key in enabling millimeter wave and terahertz radios to achieve their capacity. Due to the use of large arrays in these systems, the common exhaustive beam scanning results in a substantial training overhead. Prior work has addressed this issue, by developing compressive sensing (CS)-based methods which exploit channel sparsity to achieve faster beam alignment. Unfortunately, standard CS techniques employ wide beams and suffer from a low signal-to-noise ratio (SNR) in the channel measurements. To solve this challenge, we develop an IEEE 802.11ad/ay compatible technique that takes an in-sector approach for CS. In our method, the angle domain channel is partitioned into several sectors, and the channel within the best sector is estimated for beam alignment. The essence of our framework lies in the construction of a low-resolution beam codebook to identify the best sector and in the design of the CS matrix for in-sector channel estimation. Our beam codebook illuminates distinct non-overlapping sectors and can be realized with low-resolution phased arrays. We show that the proposed codebook results in a higher received SNR than the state-of-the-art sector sweep codebooks. Furthermore, our optimized CS matrix achieves a better in-sector channel reconstruction than comparable benchmarks.
\end{abstract}
\IEEEpeerreviewmaketitle

\section{Introduction}
\par \lettrine{M}{illimeter} wave (mmWave) and terahertz (THz) radios will be an integral part of 6G, as they can exploit wide bandwidths and achieve high beamforming gain with large antenna arrays \cite{rappaport2019wireless,sarieddeen2021overview}. 
Finding the optimal beam at the radios through classical channel estimation or beam scanning is challenging in these systems, due to the use of large arrays and hardware constraints. For instance, mmWave radios usually employ phased arrays equipped with a single radio-frequency chain, to achieve a cost-effective and power-efficient architecture \cite{roberts2021millimeter,xiao2016hierarchical}. These arrays acquire fewer spatial channel measurements than fully digital arrays, thereby restricting the application of conventional multiple-input multiple-output channel estimation methods. Furthermore, the acquired measurements are projections of the channel on antenna weight matrices whose entries are constrained to a small alphabet. The size of this alphabet is limited due to the limited resolution of the phase shifters. 

\par An alternative to beam scanning is hierarchical beam search wherein beams with decreasing beamwidths are sequentially applied. At the end of the search process, a narrow beam that results in the highest received SNR is used for communication \cite{xiao2016hierarchical}. 
To implement hierarchical search, the design objective is to construct hardware compatible beams that focus the transmitter's energy on a section of the channel in the angle domain. The sections are referred to as \textit{sectors}, where each sector represents a group of directions in the angle domain, also called the beamspace \cite{brady2013beamspace}. The process of identifying the best sector is called sector level sweep (SLS). Prior work has designed codebooks for SLS with high-resolution phased arrays \cite{xiao2016hierarchical,song2017multiresolution,xiao2018enhanced} and low-resolution phased arrays \cite{raviteja2017analog,tsai2018structured}. The techniques in \cite{xiao2016hierarchical,song2017multiresolution,xiao2018enhanced, raviteja2017analog, tsai2018structured} construct sectors that cover a contiguous set of directions. In this paper, we design a different codebook wherein the sectors cover a discontiguous set of directions.  

\par In IEEE 802.11ad/ay devices, the transmitter discerns the \textit{sector of interest} by finding the sector that results in the highest received power. Then, the beam refinement protocol (BRP) can be used to obtain channel measurements within the sector of interest and subsequently perform beam alignment. Unfortunately, exhaustive beam scanning within this sector can still result in a huge overhead in typical mmWave or THz radios. To address this problem, prior work has exploited the sparse or low-rank characteristics of the angle domain channel \cite{rangan2017rank,alkhateeb2015compressed} for sub-Nyquist channel estimation. 
Most of these techniques \cite{myers2019falp,Marzi2016Compressive,Vlachos2018massive, ma2020sparse}, however, employ quasi-omnidirectional beams to acquire spatial channel measurements. The use of such wide beams results in a poor signal-to-noise ratio (SNR) in the channel measurements, causing these methods to fail in practice \cite{tsai2018structured}. Therefore, beams that focus energy only within the sector of interest must be developed for ``in-sector'' channel estimation within BRP.  
\par We now discuss prior work \cite{tsai2018structured,wang2021jittering,ali2017millimeter,chen2020convolutional,chen2022hybrid} on sparsity-aware channel estimation within a sector. In \cite{tsai2018structured}, randomly selected discrete Fourier transform (DFT) columns were used to modulate a spread sequence and generate different beams to obtain channel measurements in different sectors. In \cite{wang2021jittering}, a contiguous band of directions is illuminated by sub-arrays. Then, the beam at each sub-array is phase modulated to construct an ensemble of in-sector beams. In \cite{ali2017millimeter}, the beams were sampled from a large set of random codes, based on their capability to focus energy on a sector of interest. We will show, by simulations, that the techniques in \cite{wang2021jittering, tsai2018structured, ali2017millimeter} result in poor in-sector channel estimates than our method. In \cite{chen2020convolutional}, linear convolution at the output of a fully digital array is performed and the steady-state output of the digital array is decimated to reduce the number of measurements and achieve parallel processing. 
Our method, unlike \cite{chen2020convolutional}, performs circular convolution that preserves the channel dimension and obtains compressive measurements of the channel without needing to discard measurements.
\par In this paper, we propose a new approach for in-sector channel estimation by solving two main challenges. First, how to design beams that focus the transmitted energy on distinct non-overlapping sectors within the beamspace. We will see that this problem is extremely challenging with low-resolution phased arrays. Second, how to optimize the compressed sensing (CS) matrices to achieve low CS aliasing artifacts in the in-sector channel estimate. We solve both these problems by constructing new in-sector beams, such that the corresponding antenna weight matrices (AWMs) can be applied in low-resolution phased arrays. These AWMs, referred to as base AWMs, each focus power on a distinct sector. In our method, the channel measurements for CS within the sector of interest are then obtained during BRP by circularly shifting the base AWM. Our main contributions are listed below.
\begin{itemize}
    \item 
    We propose a beam codebook for SLS that results in a higher received power than existing codebooks. The designed beams have a comb-like structure and illuminate distinct non-overlapping sectors in the beamspace. The AWMs associated with our beams can be realized with phase shifters with a resolution that scales logarithmically with the number of sectors. 
    \item We optimize the sequence of circular shifts applied to the base AWM to acquire channel measurements within the sector of interest. We show that the optimized shifts result in lower aliasing artifacts in CS when compared to the use of random circular shifts.
    \item We derive support recovery and in-sector channel reconstruction guarantees for CS using the orthogonal matching pursuit (OMP) algorithm. Our guarantees quantify how imperfections in the designed sectors impact the in-sector channel estimate. We show, by simulations, that our in-sector CS technique results in a lower channel reconstruction error and a higher achievable rate than comparable benchmarks.
\end{itemize}
\par The techniques in \cite{myers2019falp,tsai2018structured,myers2018spatial} and our method fall within the umbrella of convolutional compressed sensing (CCS) \cite{krahmer2014structured}, where circular shifts of a code are employed to acquire CS measurements. We now discuss how our paper distinguishes itself from prior work on CCS-based channel estimation or beam alignment. Prior work in \cite{myers2019falp} and \cite{myers2018spatial} employed perfect binary arrays and Zadoff-Chu codes as base AWMs, and used random circular shifts of these AWMs to acquire CS measurements. The AWMs in \cite{myers2018spatial,myers2019falp} generate quasi-omnidirectional beams, which result in poor SNR in the channel measurements. Furthermore, the random circular shifts used in \cite{myers2019falp, myers2018spatial} result in uniformly spread aliasing artifacts, which is also undesirable when estimating the channel within a sector. 
\par In our recent paper \cite{masoumi2022structured}, we developed an in-sector CCS-based channel estimation technique, assuming a uniform linear array at the transmitter. We assumed infinite resolution amplitude and phase control in \cite{masoumi2022structured} to construct contiguous sectors. We also optimized the circular shifts to push aliasing artifacts outside the sector of interest. Our solution in \cite{masoumi2022structured}, however, cannot be realized in practical phased arrays which require the AWM entries to be unit modulus complex exponentials with a quantized phase. In this paper, we address the practical limitations in \cite{masoumi2022structured} to propose a new class of discontiguous sectors and a new set of circular shifts.
\par \textbf{Notation}: $a$, $\mathbf{a}$ and $\mathbf{A}$ denote a scalar, vector, and a matrix. The indexing of vectors and matrices begins at $0$. $\mathbf{a}_i$ is the $i^{\mathrm{th}}$ column of $\mathbf{A}$. We denote the $(i,j)^{\mathrm{th}}$ entry of $\mathbf{A}$ by $A_{ij}$ or $A(i,j)$. The $\ell_2$ norm of $\mathbf{a}$ is denoted by $\Vert \mathbf{a} \Vert_2$. $|\mathbf{A}|$ is a matrix that contains the element-wise magnitudes of $\mathbf{A}$. The Frobenius norm of $\mathbf{A}$ is denoted by $\|\mathbf{A}\|_{\mathrm{F}}$. For an integer $N$, we define the set $[N]=\{0,\hdots,N-1\}$. Also, $\langle i\rangle_{N}$ is the modulo-$N$ remainder of $i$. We use 
$(\cdot)^{\mathrm{T}}$, $(\cdot)^{\mathrm{c}}$ and $(\cdot)^{\ast}$ to denote the transpose, conjugate and conjugate-transpose operators. We define the vector version of $\mathbf{A}$ as $\mathrm{vec}(\mathbf{A})=\left[\mathbf{a}_0^{\mathrm{T}},\cdots,\mathbf{a}_{N-1}^{\mathrm{T}}\right]^{\mathrm{T}}$. $\mathrm{Diag}(\mathbf{a})$ is a diagonal matrix with $\mathbf{a}$ on the diagonal. We use $\circledast$, $\otimes$, and $\odot$ to denote the circular convolution, the Kronecker product, and the Hadamard product. The inner product of $\mathbf{A}$ and $\mathbf{B}$ is $\langle \mathbf{A},\mathbf{B}\rangle =\sum_{i,j}\mathbf{A}_{ij}{\mathbf{B}^{\text{c}}_{ij}}$. $\mathcal{CN}(0,\sigma^{2})$ is the zero-mean complex Gaussian distribution with variance $\sigma^{2}$. $\mathsf{j}=\sqrt{-1}$. 
\section{Channel and system model}\label{sec2}
In this section, we consider a narrowband point-to-point wireless system to explain the transmit beam alignment problem. In section \ref{sec5sims}, we extend our sparse in-sector channel estimation-aided beam alignment technique to a wideband scenario using the IEEE 802.11ad frame structure.
\subsection{Channel model}\label{sec2_1sysmodel}
We consider an $N\times N$ half-wavelength spaced uniform planar array (UPA) at the transmitter (TX) and a single antenna receiver (RX) as shown in Fig.~\ref{fig_1:sysmdl}\textcolor{red}{(a)}. We assume isotropic antenna elements at the TX and the RX, for simplicity. Our solution can also be applied to other array geometries, by using appropriate array response vectors in the channel model.

\par The narrowband multiple-input single-output (MISO) channel between the TX and RX is modeled as an $N\times N$ matrix $\mathbf{H}$, the matrix representation of the $N^2\times 1$ vectorized channel. The channel comprises $L$ propagation rays, where the $\ell^{\mathrm{th}}$ ray has a complex gain of $\beta_\ell$, an azimuth angle-of-departure (AoD) $\theta_{a,\ell}$ and an elevation AoD $\theta_{e,\ell}$. By defining the beamspace angles as $\omega_{a,\ell} = \pi\sin{\theta_{e,\ell}}\sin{\theta_{a,\ell}}$, $\omega_{e,\ell} = \pi\sin{\theta_{e,\ell}}\cos{\theta_{a,\ell}}$ \cite{tse2005fundamentals} and the $N\times 1$ Vandermonde vector $\mathbf{a}_{N}(\omega)$ as
\begin{equation}\label{eqn:vandrmond}
    \mathbf{a}_{N}(\omega)=\left[1, e^{\mathrm{j} \omega}, e^{\mathrm{j} 2 \omega}, \cdots, e^{\mathrm{j}(N-1) \omega}\right]^\mathrm{T},
\end{equation}
the baseband channel matrix $\mathbf{H}$ is given by
\begin{equation}\label{eqn:channel}
    \mathbf{H}=\sum_{\ell=1}^{L} \beta_\ell \mathbf{a}_{N}\left(\omega_{e, \ell}\right) \mathbf{a}_{N}^\mathrm{T}\left(\omega_{a, \ell}\right).
\end{equation}
The channel dimension $N^2$ can be in the order of hundreds to thousands in typical mmWave or THz access points. 
\par The channel $\mathbf{H}$ is approximately sparse in the angle domain due to high scattering at mmWave and THz wavelengths \cite{sarieddeen2021overview}. To exploit this property during channel reconstruction, $\mathbf{H}$ is expressed in the beamspace (angle domain) where it has a sparse representation. Since we assume a UPA at the TX, the 2D-DFT dictionary is used to represent $\mathbf{H}$ in the beamspace. We use $\mathbf{U}_{N}$ to denote the standard $N\times N$ unitary DFT matrix and $\mathbf{X}$ to denote the beamspace representation of $\mathbf{H}$. Then, $\mathbf{X}$  and $\mathbf{H}$ are related as
\begin{equation}\label{eqn:beam_H}
    \mathbf{H} = \mathbf{U}_{N}\mathbf{X}\mathbf{U}_{N}.
\end{equation}
In our analysis, we assume that $\mathbf{X}$ is exactly sparse, i.e., the beamspace angles are exactly aligned with any of the $N^2$ directional 2D-DFT beam directions. In our simulations, we use channels obtained from the NYU channel simulator \cite{sun2017novel} where $\mathbf{X}$ is only approximately sparse.
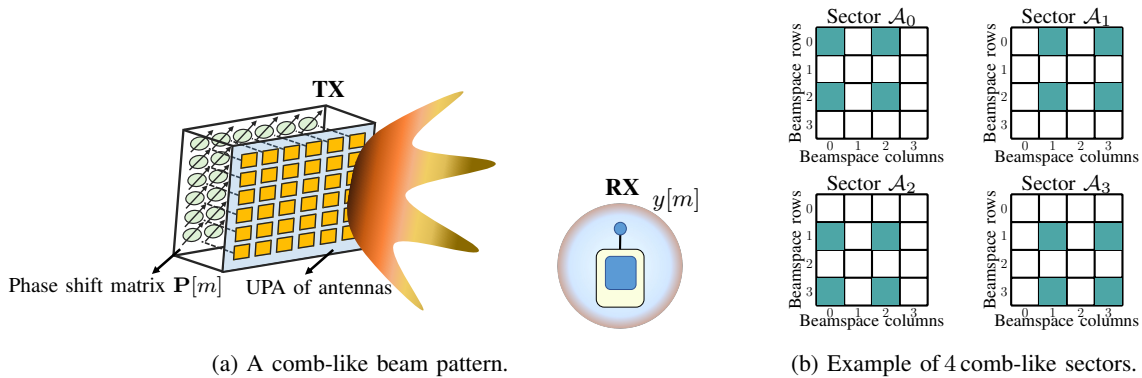
\begin{figure}[t]
\centering
\subfloat[A comb-like beam pattern.]{\begin{tikzpicture}[scale=0.36]
\node[draw=none,fill=none] at (-1,0){\includegraphics[trim=6cm 5.85cm 6.8cm 3.8cm, clip, width=0.45\textwidth]{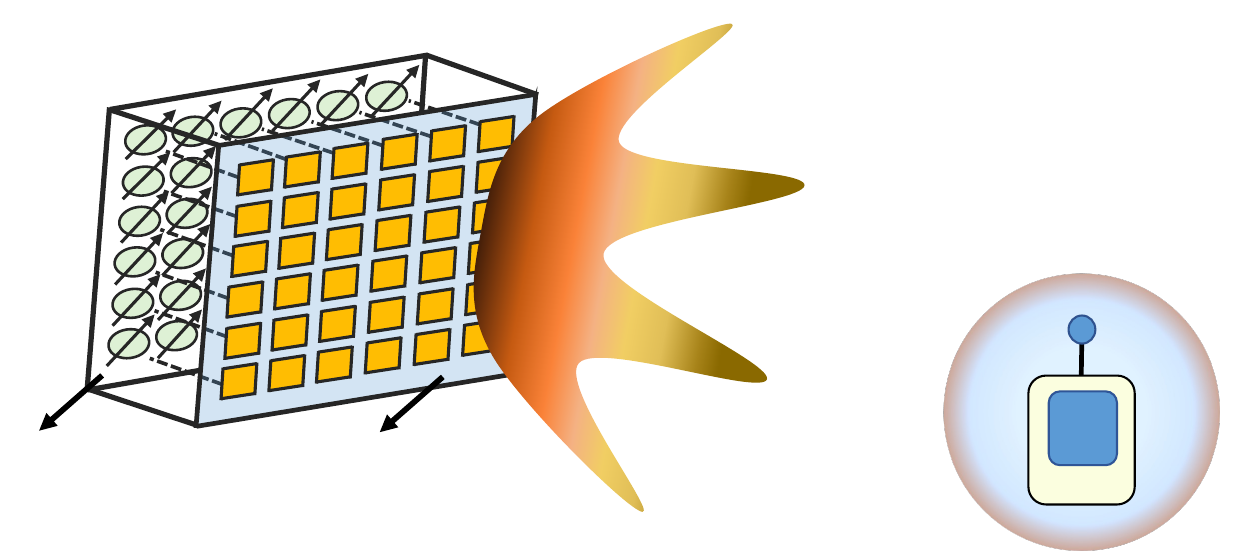}};
\node[draw=none,fill=none] at (-12,-3) {\scriptsize Phase shift matrix $\mathbf{P}[m]$};
\node[draw=none,fill=none] at (-4.5,-3) {\scriptsize UPA of antennas};
\node[draw=none,fill=none] at (-4.15,4.3) {\footnotesize\textbf{TX}};
\node[draw=none,fill=none] at (6.7,0.7) {\footnotesize\textbf{RX}};
\node[draw=none,fill=none] at (8.7,0.2) {\footnotesize $y[m]$};
\end{tikzpicture}
}\label{fig_1a:sysmdl}
\hfil
\subfloat[Example of $\!4\!$ comb-like sectors.]{\begin{tikzpicture}[scale=0.37, every node/.style={minimum size=0.37cm-\pgflinewidth, outer sep=0pt}]
\node at (1, 3.37)   (a) {\scriptsize Sector $\mathcal{A}_2$};
\draw[step=1cm,black,thick] (-1,-1) grid (3,3);
\node[fill=teal!70] at (-0.5,1.5) {};
\node[fill=teal!70] at (1.5,1.5) {};
\node[fill=teal!70] at (-0.5,-0.5) {};
\node[fill=teal!70] at (1.5,-0.5) {};
\node at (-1.25,-0.5) {\scalebox{.43}{$3$}};
\node at (-1.25,0.5) {\scalebox{.43}{$2$}};
\node at (-1.25,1.5) {\scalebox{.43}{$1$}};
\node at (-1.25,2.5) {\scalebox{.43}{$0$}};
\node at (-0.5,-1.25) {\scalebox{.43}{$0$}};
\node at (0.5,-1.25) {\scalebox{.43}{$1$}};
\node at (1.5,-1.25) {\scalebox{.43}{$2$}};
\node at (2.5,-1.25) {\scalebox{.43}{$3$}};
\node[rotate=90] at (-1.7, 1)   (a) {\scalebox{.55}{Beamspace rows}};
\node at (1, -1.6)   (a) {\scalebox{.55}{Beamspace columns}};
\node at (8, 3.37)   (a) {\scriptsize Sector $\mathcal{A}_3$};
\draw[step=1cm,black,thick] (6,-1) grid (10,3);
\node[fill=teal!70] at (7.5,1.5) {};
\node[fill=teal!70] at (9.5,1.5) {};
\node[fill=teal!70] at (7.5,-0.5) {};
\node[fill=teal!70] at (9.5,-0.5) {};
\node at (5.75,-0.5) {\scalebox{.43}{$3$}};
\node at (5.75,0.5) {\scalebox{.43}{$2$}};
\node at (5.75,1.5) {\scalebox{.43}{$1$}};
\node at (5.75,2.5) {\scalebox{.43}{$0$}};
\node at (6.5,-1.25) {\scalebox{.43}{$0$}};
\node at (7.5,-1.25) {\scalebox{.43}{$1$}};
\node at (8.5,-1.25) {\scalebox{.43}{$2$}};
\node at (9.5,-1.25) {\scalebox{.43}{$3$}};
\node[rotate=90] at (5.3, 1)   (a) {\scalebox{.55}{Beamspace rows}};
\node at (8, -1.6)   (a) {\scalebox{.55}{Beamspace columns}};
\node at (1, 9.37)   (a) {\scriptsize Sector $\mathcal{A}_0$};
\draw[step=1cm,black,thick] (-1,5) grid (3,9);
\node[fill=teal!70] at (-0.5,8.5) {};
\node[fill=teal!70] at (1.5,8.5) {};
\node[fill=teal!70] at (-0.5,6.5) {};
\node[fill=teal!70] at (1.5,6.5) {};
\node at (-1.25,5.5) {\scalebox{.43}{$3$}};
\node at (-1.25,6.5) {\scalebox{.43}{$2$}};
\node at (-1.25,7.5) {\scalebox{.43}{$1$}};
\node at (-1.25,8.5) {\scalebox{.43}{$0$}};
\node at (-0.5,4.75) {\scalebox{.43}{$0$}};
\node at (0.5,4.75) {\scalebox{.43}{$1$}};
\node at (1.5,4.75) {\scalebox{.43}{$2$}};
\node at (2.5,4.75) {\scalebox{.43}{$3$}};
\node[rotate=90] at (-1.7, 7)   (a) {\scalebox{.55}{Beamspace rows}};
\node at (1, 4.3)   (a) {\scalebox{.55}{Beamspace columns}};
\node at (8, 9.37)   (a) {\scriptsize Sector $\mathcal{A}_1$};
\draw[step=1cm,black,thick] (6,5) grid (10,9);
\node[fill=teal!70] at (7.5,8.5) {};
\node[fill=teal!70] at (9.5,8.5) {};
\node[fill=teal!70] at (7.5,6.5) {};
\node[fill=teal!70] at (9.5,6.5) {};
\node at (5.75,5.5) {\scalebox{.43}{$3$}};
\node at (5.75,6.5) {\scalebox{.43}{$2$}};
\node at (5.75,7.5) {\scalebox{.43}{$1$}};
\node at (5.75,8.5) {\scalebox{.43}{$0$}};
\node at (6.5,4.75) {\scalebox{.43}{$0$}};
\node at (7.5,4.75) {\scalebox{.43}{$1$}};
\node at (8.5,4.75) {\scalebox{.43}{$2$}};
\node at (9.5,4.75) {\scalebox{.43}{$3$}};
\node[rotate=90] at (5.3, 7)   (a) {\scalebox{.55}{Beamspace rows}};
\node at (8, 4.3)   (a) {\scalebox{.55}{Beamspace columns}};
\end{tikzpicture}}
\caption{\small An mmWave system with a UPA at the TX employing comb-like beams designed in this work. Each comb-like beam illuminates a beamspace sector shown in (b). In this paper, the sector that results in the highest received power is identified and then the channel within that sector is estimated for beam alignment. \normalsize}\label{fig_1:sysmdl} 
\end{figure}
\subsection{System model}\label{sec2B}
\par The TX uses a phased array wherein the antennas are connected to a single radio frequency chain through a $q$-bit phase shifter. The set of the possible weights at each antenna is $\mathbb{Q}_{q} = \{e^{\mathsf{j}2\pi\ell/2^{q}}/N:\ell\in [2^q]\}$. The TX can thus apply AWMs to its phased array, which are constrained to be within $\mathbb{Q}_q^{N\times N}$. In this paper, the AWMs are designed such that their corresponding beams illuminate only a certain set of indices of the beamspace $\mathbf{X}$. A \textit{sector} is defined as a group of indices of the beamspace that are illuminated by a beam. We use $S$ to denote the number of sectors used to partition the $N\times N$ beamspace grid into disjoint subsets $\mathcal{A}_s,~\forall s\in[S]$. The union of $\{\mathcal{A}_s\}_{s=0}^{S-1}$, i.e., $[N]\times [N]$, spans the entire beamspace. An example of the sectors $\{\mathcal{A}_s\}_{s=0}^{S-1}$ with our comb-like construction proposed in Sec. \ref{sec3_sls_awm} is shown in Fig. \ref{fig_1:sysmdl}\textcolor{red}{(b)} for $S=4$. 
\par We now explain how our in-sector CS-based beam alignment method can be integrated into \textit{SLS} and \textit{BRP} of the beam training procedure in IEEE 802.11ad/ay \cite{nitsche2014ieee} standards.
\subsubsection{Sector level sweep}
During SLS, the TX uses sector sweep frames and illuminates the channel with sectored beam patterns. The standard, however, does not prescribe specific beam patterns to conduct SLS. These beams may be the common contiguous patterns discussed in \cite{tsai2018structured,xiao2018enhanced,raviteja2017analog} or the comb-like patterns proposed in this work as shown in Fig.~\ref{fig_1:sysmdl}\textcolor{red}{(b)}. We define $\{\mathbf{P}_s\}_{s=0}^{S-1}$ as the base AWMs applied at the TX for SLS. 
We use $y^{\mathrm{sls}}_s$ to denote the measurement acquired when the TX applies $\mathbf{P}_s$ and $v_s\sim\mathcal{CN}(0,\sigma^2)$ as the noise in this measurement, i.e.,
\begin{equation}\label{eqn:rx_sls}
    y^{\mathrm{sls}}_s = \left\langle \mathbf{H},\mathbf{P}_s \right\rangle + v_s.
\end{equation}
 In SLS, the RX calculates the received power associated with \eqref{eqn:rx_sls} for each sector. Then, it determines the sector of interest as the one with the highest received power using $s_{\mathrm{opt}} = \underset{s\in[S]}{\mathrm{argmax}}~|y^{\mathrm{sls}}_s|^2$. An illustration of the SLS procedure is shown in Fig.~\ref{fig:sls}. 
 \par We use $\Po\in\{\mathbf{P}_s\}_{s=0}^{S-1}$ to denote the AWM associated with $s_{\mathrm{opt}}$ and $\Ao$ to denote the best sector corresponding to $\Po$. To see this correspondence, observe that the inner product of two matrices is the same as the inner product of their inverse 2D-DFTs, i.e., $\left\langle \mathbf{H},\mathbf{P}_s \right\rangle=\left\langle \mathbf{X},\mathbf{U}^{\ast}_N \mathbf{P}_s \mathbf{U}^{\ast}_N \right\rangle$. As a result, the TX generates a beam that illuminates the beamspace indices at which $\mathbf{U}^{\ast}_N \mathbf{P}_s \mathbf{U}^{\ast}_N$ is non-zero, when it applies $\mathbf{P}_s$. Therefore, $\mathcal{A}_s=\{(k,\ell): |(\mathbf{U}^{\ast}_N \mathbf{P}_s \mathbf{U}^{\ast}_N)_{k,\ell}\neq 0|\}$. Our objective in codebook design for SLS is to design $\{\mathbf{P}_s\}_{s=0}^{S-1}$ such that each of them illuminates a distinct set of indices in the beamspace, i.e., $\{\mathcal{A}_s\}_{s=0}^{S-1}$ are disjoint and their union is the entire beamspace. The main challenge for this design is due to the constraint that $\mathbf{P}_s\in\mathbb{Q}_{q}^{N\times N} \forall s$. In Section \ref{sec3_sls_awm}, we present our codebook and show that it outperforms the codebooks in \cite{tsai2018structured,xiao2018enhanced,raviteja2017analog}.
\begin{figure}[t]
\centering
\includegraphics[trim=3.4cm 8.1cm 3.25cm 7.3cm, clip, scale=0.5]{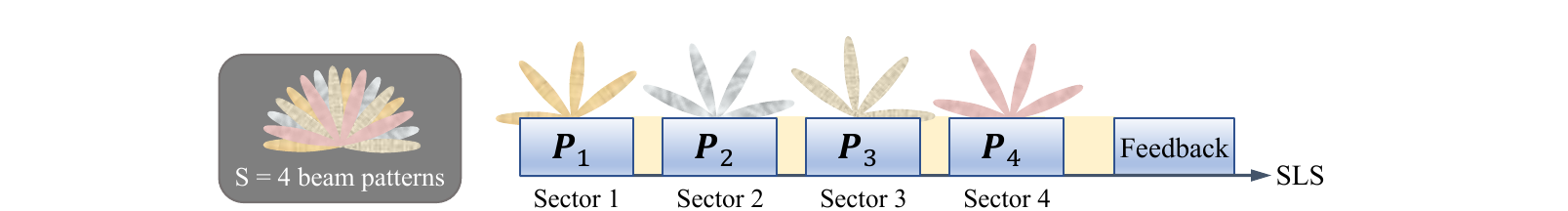}
\caption{\small Example of SLS to determine the sector of interest using $S=4$ disjoint comb-like sectors designed with our approach. The RX finds the sector of interest $\Ao$ as the one that results in the highest received power and feedbacks this information to the TX. The TX can then estimate the channel within the sector using BRP.\normalsize}\label{fig:sls} 
\end{figure}
\subsubsection{Beam refinement protocol}
\par During BRP, the TX applies a collection of AWMs that illuminate directions within the sector of interest. Under the assumption that the $S$ sectors uniformly partition the beamspace, the in-sector channel is an $N^2/S$ dimension vector. A simple approach to estimate this vector is to apply directional beams that exhaustively scan all directions within the sector. Such a scan, however, results in a substantial training overhead, which is in the order of $N^2/S$. To avoid this high training overhead, we develop an in-sector CS-based method that obtains compressive measurements of the channel within the sector of interest. 
\par We use $M$ to denote the number of CS measurements within the sector of interest $\Ao$ to perform in-sector CS. These measurements are obtained using an ensemble of beams $\{\Po[m]\}_{m=0}^{M-1}$ that only illuminate the sector of interest, i.e., the matrices $\{\mathbf{U}^{\ast}_N \Po[m] \mathbf{U}^{\ast}_N\}_{m=0}^{M-1}$ are non-zero only at the indices in $\Ao$. For a measurement noise $v[m]\sim\mathcal{CN}(0,\sigma^2)$, the $m^{\mathrm{th}}$ in-sector measurement $y[m]$ of the channel is
\begin{equation}\label{eqn:rx_incs}
    y[m] = \left\langle \mathbf{H},\Po[m] \right\rangle + v[m].
\end{equation}
Our goal for in-sector CS is to construct $\{\Po[m]\}_{m=0}^{M-1}$ such that: i) $\Po[m]\in\mathbb{Q}_{q}^{N\times N}$ illuminates only the directions within $\Ao$ for each $m$ and ii) the collection results in a low channel estimation error with CS algorithms. In-sector CS is promising over standard CS techniques that employ wide beams, because the SNR of the in-sector measurements is higher than those acquired with a wide beam by about $10\,\mathrm{log}_{10}S$ dB.
\section{Proposed sectors and base AWMs for SLS}\label{sec3_sls_awm}
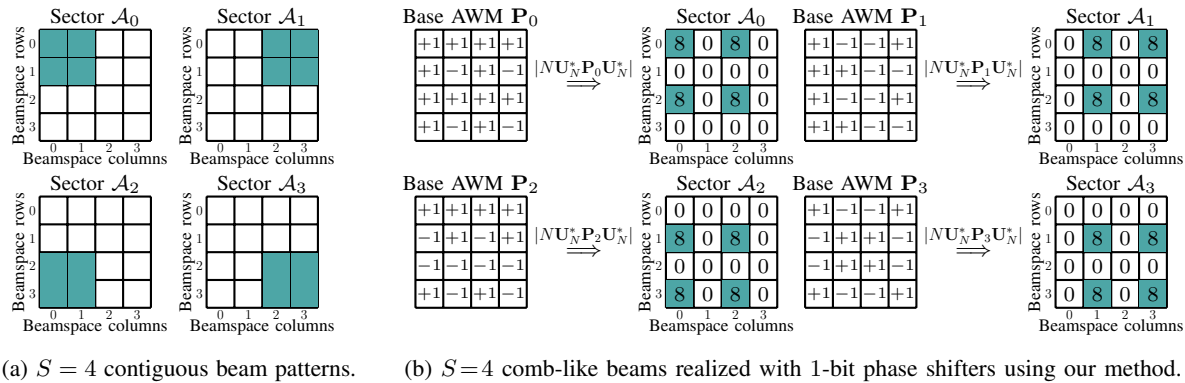
\begin{figure}[t]
\centering
\subfloat[$S=4$ contiguous beam patterns.]{\begin{tikzpicture}[scale=0.37, every node/.style={minimum size=0.37cm-\pgflinewidth, outer sep=0pt}]
\node at (2, 3.45)   (a) {\scriptsize Sector $\mathcal{A}_2$};
\draw[step=1cm,black,thick] (0,-1) grid (4,3);
\node[fill=teal!70] at (0.5,0.5) {};
\node[fill=teal!70] at (1.5,0.5) {};
\node[fill=teal!70] at (0.5,-0.5) {};
\node[fill=teal!70] at (1.5,-0.5) {};
\node at (-0.25,-0.5) {\scalebox{.43}{$3$}};
\node at (-0.25,0.5) {\scalebox{.43}{$2$}};
\node at (-0.25,1.5) {\scalebox{.43}{$1$}};
\node at (-0.25,2.5) {\scalebox{.43}{$0$}};
\node at (0.5,-1.25) {\scalebox{.43}{$0$}};
\node at (1.5,-1.25) {\scalebox{.43}{$1$}};
\node at (2.5,-1.25) {\scalebox{.43}{$2$}};
\node at (3.5,-1.25) {\scalebox{.43}{$3$}};
\node[rotate=90] at (-0.7, 1)   (a) {\scalebox{.55}{Beamspace rows}};
\node at (2, -1.7)   (a) {\scalebox{.55}{Beamspace columns}};
\node at (8, 3.45)   (a) {\scriptsize Sector $\mathcal{A}_3$};
\draw[step=1cm,black,thick] (6,-1) grid (10,3);
\node[fill=teal!70] at (8.5,0.5) {};
\node[fill=teal!70] at (9.5,0.5) {};
\node[fill=teal!70] at (8.5,-0.5) {};
\node[fill=teal!70] at (9.5,-0.5) {};
\node at (5.75,-0.5) {\scalebox{.43}{$3$}};
\node at (5.75,0.5) {\scalebox{.43}{$2$}};
\node at (5.75,1.5) {\scalebox{.43}{$1$}};
\node at (5.75,2.5) {\scalebox{.43}{$0$}};
\node at (6.5,-1.25) {\scalebox{.43}{$0$}};
\node at (7.5,-1.25) {\scalebox{.43}{$1$}};
\node at (8.5,-1.25) {\scalebox{.43}{$2$}};
\node at (9.5,-1.25) {\scalebox{.43}{$3$}};
\node[rotate=90] at (5.3, 1)   (a) {\scalebox{.55}{Beamspace rows}};
\node at (8, -1.7)   (a) {\scalebox{.55}{Beamspace columns}};
\node at (2, 9.45)   (a) {\scriptsize Sector $\mathcal{A}_0$};
\draw[step=1cm,black,thick] (0,5) grid (4,9);
\node[fill=teal!70] at (0.5,8.5) {};
\node[fill=teal!70] at (1.5,8.5) {};
\node[fill=teal!70] at (0.5,7.5) {};
\node[fill=teal!70] at (1.5,7.5) {};
\node at (-0.25,5.5) {\scalebox{.43}{$3$}};
\node at (-0.25,6.5) {\scalebox{.43}{$2$}};
\node at (-0.25,7.5) {\scalebox{.43}{$1$}};
\node at (-0.25,8.5) {\scalebox{.43}{$0$}};
\node at (0.5,4.75) {\scalebox{.43}{$0$}};
\node at (1.5,4.75) {\scalebox{.43}{$1$}};
\node at (2.5,4.75) {\scalebox{.43}{$2$}};
\node at (3.5,4.75) {\scalebox{.43}{$3$}};
\node[rotate=90] at (-0.7, 7)   (a) {\scalebox{.55}{Beamspace rows}};
\node at (2, 4.3)   (a) {\scalebox{.55}{Beamspace columns}};
\node at (8, 9.45)   (a) {\scriptsize Sector $\mathcal{A}_1$};
\draw[step=1cm,black,thick] (6,5) grid (10,9);
\node[fill=teal!70] at (8.5,8.5) {};
\node[fill=teal!70] at (9.5,8.5) {};
\node[fill=teal!70] at (8.5,7.5) {};
\node[fill=teal!70] at (9.5,7.5) {};
\node at (5.75,5.5) {\scalebox{.43}{$3$}};
\node at (5.75,6.5) {\scalebox{.43}{$2$}};
\node at (5.75,7.5) {\scalebox{.43}{$1$}};
\node at (5.75,8.5) {\scalebox{.43}{$0$}};
\node at (6.5,4.75) {\scalebox{.43}{$0$}};
\node at (7.5,4.75) {\scalebox{.43}{$1$}};
\node at (8.5,4.75) {\scalebox{.43}{$2$}};
\node at (9.5,4.75) {\scalebox{.43}{$3$}};
\node[rotate=90] at (5.3, 7)   (a) {\scalebox{.55}{Beamspace rows}};
\node at (8, 4.3)   (a) {\scalebox{.55}{Beamspace columns}};
\end{tikzpicture}
}\label{fig_3a:contiguous}
\hfil
\subfloat[$S\!=\!4$ comb-like beams realized with $1$-bit phase shifters using our method.]{\begin{tikzpicture}[scale=0.37, every node/.style={minimum size=0.37cm-\pgflinewidth, outer sep=0pt}]
\node at (-4, 3.45)   (a) {\scriptsize Sector $\mathcal{A}_2$};
\draw[step=1cm,black,thick] (-6,-1) grid (-2,3);
\node[fill=teal!70] at (-5.5,1.5) {};
\node[fill=teal!70] at (-3.5,1.5) {};
\node[fill=teal!70] at (-5.5,-0.5) {};
\node[fill=teal!70] at (-3.5,-0.5) {};
\node[rotate=90] at (-6.7, 1)   (a) {\scalebox{.55}{Beamspace rows}};
\node at (-4, -1.7)   (a) {\scalebox{.55}{Beamspace columns}};
\node at (-5.5,2.5) {\scriptsize$0$};
\node at (-5.5,1.5) {\scriptsize$8$};
\node at (-5.5,0.5) {\scriptsize$0$};
\node at (-5.5,-0.5) {\scriptsize$8$};
\node at (-4.5,2.5) {\scriptsize$0$};
\node at (-4.5,1.5) {\scriptsize$0$};
\node at (-4.5,0.5) {\scriptsize$0$};
\node at (-4.5,-0.5) {\scriptsize$0$};
\node at (-3.5,2.5) {\scriptsize$0$};
\node at (-3.5,1.5) {\scriptsize$8$};
\node at (-3.5,0.5) {\scriptsize$0$};
\node at (-3.5,-0.5) {\scriptsize$8$};
\node at (-2.5,2.5) {\scriptsize$0$};
\node at (-2.5,1.5) {\scriptsize$0$};
\node at (-2.5,0.5) {\scriptsize$0$};
\node at (-2.5,-0.5) {\scriptsize$0$};
\node at (-6.25,-0.5) {\scalebox{.43}{$3$}};
\node at (-6.25,0.5) {\scalebox{.43}{$2$}};
\node at (-6.25,1.5) {\scalebox{.43}{$1$}};
\node at (-6.25,2.5) {\scalebox{.43}{$0$}};
\node at (-5.5,-1.25) {\scalebox{.43}{$0$}};
\node at (-4.5,-1.25) {\scalebox{.43}{$1$}};
\node at (-3.5,-1.25) {\scalebox{.43}{$2$}};
\node at (-2.5,-1.25) {\scalebox{.43}{$3$}};
\node at (-13, 3.45)   (a) {\scriptsize Base AWM $\mathbf{P}_2$};
\draw[step=1cm,black,thick] (-15,-1) grid (-11,3);
\node at (-14.5,2.5) {\tiny$+1$};
\node at (-14.5,1.5) {\tiny$-1$};
\node at (-14.5,0.5) {\tiny$-1$};
\node at (-14.5,-0.5) {\tiny$+1$};
\node at (-13.5,2.5) {\tiny$+1$};
\node at (-13.5,1.5) {\tiny$+1$};
\node at (-13.5,0.5) {\tiny$-1$};
\node at (-13.5,-0.5) {\tiny$-1$};
\node at (-12.5,2.5) {\tiny$+1$};
\node at (-12.5,1.5) {\tiny$-1$};
\node at (-12.5,0.5) {\tiny$-1$};
\node at (-12.5,-0.5) {\tiny$+1$};
\node at (-11.5,2.5) {\tiny$+1$};
\node at (-11.5,1.5) {\tiny$+1$};
\node at (-11.5,0.5) {\tiny$-1$};
\node at (-11.5,-0.5) {\tiny$-1$};
\node at (-8.95,1) {\scriptsize{$\Longrightarrow$}};
\node at (-8.95,1.6) {\scalebox{0.55}{$|N\mathbf{U}_{N}^{*}\mathbf{P}_2\mathbf{U}_{N}^{*}|$}};
\node at (10, 3.45)   (a) {\scriptsize Sector $\mathcal{A}_3$};
\draw[step=1cm,black,thick] (8,-1) grid (12,3);
\node[fill=teal!70] at (9.5,1.5) {};
\node[fill=teal!70] at (11.5,1.5) {};
\node[fill=teal!70] at (9.5,-0.5) {};
\node[fill=teal!70] at (11.5,-0.5) {};
\node[rotate=90] at (7.3, 1)   (a) {\scalebox{.55}{Beamspace rows}};
\node at (10, -1.7)   (a) {\scalebox{.55}{Beamspace columns}};
\node at (8.5,2.5) {\scriptsize$0$};
\node at (8.5,1.5) {\scriptsize$0$};
\node at (8.5,0.5) {\scriptsize$0$};
\node at (8.5,-0.5) {\scriptsize$0$};
\node at (9.5,2.5) {\scriptsize$0$};
\node at (9.5,1.5) {\scriptsize$8$};
\node at (9.5,0.5) {\scriptsize$0$};
\node at (9.5,-0.5) {\scriptsize$8$};
\node at (10.5,2.5) {\scriptsize$0$};
\node at (10.5,1.5) {\scriptsize$0$};
\node at (10.5,0.5) {\scriptsize$0$};
\node at (10.5,-0.5) {\scriptsize$0$};
\node at (11.5,2.5) {\scriptsize$0$};
\node at (11.5,1.5) {\scriptsize$8$};
\node at (11.5,0.5) {\scriptsize$0$};
\node at (11.5,-0.5) {\scriptsize$8$};
\node at (7.75,-0.5) {\scalebox{.43}{$3$}};
\node at (7.75,0.5) {\scalebox{.43}{$2$}};
\node at (7.75,1.5) {\scalebox{.43}{$1$}};
\node at (7.75,2.5) {\scalebox{.43}{$0$}};
\node at (8.5,-1.25) {\scalebox{.43}{$0$}};
\node at (9.5,-1.25) {\scalebox{.43}{$1$}};
\node at (10.5,-1.25) {\scalebox{.43}{$2$}};
\node at (11.5,-1.25) {\scalebox{.43}{$3$}};
\node at (1, 3.45)   (a) {\scriptsize Base AWM $\mathbf{P}_3$};
\draw[step=1cm,black,thick] (-1,-1) grid (3,3);
\node at (-.5,2.5) {\tiny$+1$};
\node at (-.5,1.5) {\tiny$-1$};
\node at (-.5,0.5) {\tiny$-1$};
\node at (-.5,-0.5) {\tiny$+1$};
\node at (0.5,2.5) {\tiny$-1$};
\node at (0.5,1.5) {\tiny$+1$};
\node at (0.5,0.5) {\tiny$+1$};
\node at (0.5,-0.5) {\tiny$-1$};
\node at (1.5,2.5) {\tiny$-1$};
\node at (1.5,1.5) {\tiny$+1$};
\node at (1.5,0.5) {\tiny$+1$};
\node at (1.5,-0.5) {\tiny$-1$};
\node at (2.5,2.5) {\tiny$+1$};
\node at (2.5,1.5) {\tiny$-1$};
\node at (2.5,0.5) {\tiny$-1$};
\node at (2.5,-0.5) {\tiny$+1$};
\node at (5.05,1) {\scriptsize{$\Longrightarrow$}};
\node at (5.05,1.6) {\scalebox{0.55}{$|N\mathbf{U}_{N}^{*}\mathbf{P}_3\mathbf{U}_{N}^{*}|$}};
\node at (-4, 9.45)   (a) {\scriptsize Sector $\mathcal{A}_0$};
\draw[step=1cm,black,thick] (-6,5) grid (-2,9);
\node[fill=teal!70] at (-5.5,8.5) {};
\node[fill=teal!70] at (-3.5,8.5) {};
\node[fill=teal!70] at (-5.5,6.5) {};
\node[fill=teal!70] at (-3.5,6.5) {};
\node[rotate=90] at (-6.7, 7)   (a) {\scalebox{.55}{Beamspace rows}};
\node at (-4, 4.3)   (a) {\scalebox{.55}{Beamspace columns}};
\node at (-5.5,8.5) {\scriptsize$8$};
\node at (-5.5,7.5) {\scriptsize$0$};
\node at (-5.5,6.5) {\scriptsize$8$};
\node at (-5.5,5.5) {\scriptsize$0$};
\node at (-4.5,8.5) {\scriptsize$0$};
\node at (-4.5,7.5) {\scriptsize$0$};
\node at (-4.5,6.5) {\scriptsize$0$};
\node at (-4.5,5.5) {\scriptsize$0$};
\node at (-3.5,8.5) {\scriptsize$8$};
\node at (-3.5,7.5) {\scriptsize$0$};
\node at (-3.5,6.5) {\scriptsize$8$};
\node at (-3.5,5.5) {\scriptsize$0$};
\node at (-2.5,8.5) {\scriptsize$0$};
\node at (-2.5,7.5) {\scriptsize$0$};
\node at (-2.5,6.5) {\scriptsize$0$};
\node at (-2.5,5.5) {\scriptsize$0$};
\node at (-6.25,5.5) {\scalebox{.43}{$3$}};
\node at (-6.25,6.5) {\scalebox{.43}{$2$}};
\node at (-6.25,7.5) {\scalebox{.43}{$1$}};
\node at (-6.25,8.5) {\scalebox{.43}{$0$}};
\node at (-5.5,4.75) {\scalebox{.43}{$0$}};
\node at (-4.5,4.75) {\scalebox{.43}{$1$}};
\node at (-3.5,4.75) {\scalebox{.43}{$2$}};
\node at (-2.5,4.75) {\scalebox{.43}{$3$}};
\node at (-13, 9.45)   (a) {\scriptsize Base AWM $\mathbf{P}_0$};
\draw[step=1cm,black,thick] (-15,5) grid (-11,9);
\node at (-14.5,8.5) {\tiny$+1$};
\node at (-14.5,7.5) {\tiny$+1$};
\node at (-14.5,6.5) {\tiny$+1$};
\node at (-14.5,5.5) {\tiny$+1$};
\node at (-13.5,8.5) {\tiny$+1$};
\node at (-13.5,7.5) {\tiny$-1$};
\node at (-13.5,6.5) {\tiny$+1$};
\node at (-13.5,5.5) {\tiny$-1$};
\node at (-12.5,8.5) {\tiny$+1$};
\node at (-12.5,7.5) {\tiny$+1$};
\node at (-12.5,6.5) {\tiny$+1$};
\node at (-12.5,5.5) {\tiny$+1$};
\node at (-11.5,8.5) {\tiny$+1$};
\node at (-11.5,7.5) {\tiny$-1$};
\node at (-11.5,6.5) {\tiny$+1$};
\node at (-11.5,5.5) {\tiny$-1$};
\node at (-8.95,7) {\scriptsize{$\Longrightarrow$}};
\node at (-8.95,7.6) {\scalebox{0.55}{$|N\mathbf{U}_{N}^{*}\mathbf{P}_0\mathbf{U}_{N}^{*}|$}};
\node at (10, 9.45)   (a) {\scriptsize Sector $\mathcal{A}_1$};
\draw[step=1cm,black,thick] (8,5) grid (12,9);
\node[fill=teal!70] at (9.5,8.5) {};
\node[fill=teal!70] at (11.5,8.5) {};
\node[fill=teal!70] at (9.5,6.5) {};
\node[fill=teal!70] at (11.5,6.5) {};
\node[rotate=90] at (7.3, 7)   (a) {\scalebox{.55}{Beamspace rows}};
\node at (10, 4.3)   (a) {\scalebox{.55}{Beamspace columns}};
\node at (8.5,8.5) {\scriptsize$0$};
\node at (8.5,7.5) {\scriptsize$0$};
\node at (8.5,6.5) {\scriptsize$0$};
\node at (8.5,5.5) {\scriptsize$0$};
\node at (9.5,8.5) {\scriptsize$8$};
\node at (9.5,7.5) {\scriptsize$0$};
\node at (9.5,6.5) {\scriptsize$8$};
\node at (9.5,5.5) {\scriptsize$0$};
\node at (10.5,8.5) {\scriptsize$0$};
\node at (10.5,7.5) {\scriptsize$0$};
\node at (10.5,6.5) {\scriptsize$0$};
\node at (10.5,5.5) {\scriptsize$0$};
\node at (11.5,8.5) {\scriptsize$8$};
\node at (11.5,7.5) {\scriptsize$0$};
\node at (11.5,6.5) {\scriptsize$8$};
\node at (11.5,5.5) {\scriptsize$0$};
\node at (7.75,5.5) {\scalebox{.43}{$3$}};
\node at (7.75,6.5) {\scalebox{.43}{$2$}};
\node at (7.75,7.5) {\scalebox{.43}{$1$}};
\node at (7.75,8.5) {\scalebox{.43}{$0$}};
\node at (8.5,4.75) {\scalebox{.43}{$0$}};
\node at (9.5,4.75) {\scalebox{.43}{$1$}};
\node at (10.5,4.75) {\scalebox{.43}{$2$}};
\node at (11.5,4.75) {\scalebox{.43}{$3$}};
\node at (1, 9.45)   (a) {\scriptsize Base AWM $\mathbf{P}_1$};
\draw[step=1cm,black,thick] (-1,5) grid (3,9);
\node at (-.5,8.5) {\tiny$+1$};
\node at (-.5,7.5) {\tiny$+1$};
\node at (-.5,6.5) {\tiny$+1$};
\node at (-.5,5.5) {\tiny$+1$};
\node at (0.5,8.5) {\tiny$-1$};
\node at (0.5,7.5) {\tiny$+1$};
\node at (0.5,6.5) {\tiny$-1$};
\node at (0.5,5.5) {\tiny$+1$};
\node at (1.5,8.5) {\tiny$-1$};
\node at (1.5,7.5) {\tiny$-1$};
\node at (1.5,6.5) {\tiny$-1$};
\node at (1.5,5.5) {\tiny$-1$};
\node at (2.5,8.5) {\tiny$+1$};
\node at (2.5,7.5) {\tiny$-1$};
\node at (2.5,6.5) {\tiny$+1$};
\node at (2.5,5.5) {\tiny$-1$};
\node at (5.05,7) {\scriptsize{$\Longrightarrow$}};
\node at (5.05,7.6) {\scalebox{0.55}{$|N\mathbf{U}_{N}^{*}\mathbf{P}_1\mathbf{U}_{N}^{*}|$}};
\end{tikzpicture}}
\caption{\small Contiguous and comb-like beam patterns to illuminate $S=4$ disjoint sections of the 2D beamspace with a $4\times 4$ array at the TX. With $1$-bit phase shifters, the  AWMs corresponding to the beam patterns in (a) cannot be realized, while those corresponding to the patterns in (b) can be achieved using our construction in Sec. \ref{subsec_base_awm}.\normalsize}\label{fig_2:sectorTypes} 
\end{figure}
\par SLS is usually performed using beam patterns that illuminate disjoint sectors within the beamspace. AWMs that achieve contiguous illumination patterns, such as the ones shown in Fig.~\ref{fig_2:sectorTypes}\textcolor{red}{(a)}, are common in the SLS literature \cite{giordani2016initial, xiao2018enhanced,raviteja2017analog}. Such contiguous patterns, however, cannot be realized when the TX has low-resolution phase shifters. To see this, consider the problem of constructing the illumination patterns in Fig.~\ref{fig_2:sectorTypes}\textcolor{red}{(a)} using a $4\times 4$ antenna array with one-bit phase shifters. Here, the AWMs need to be chosen from $\{\pm 1\}^{4 \times 4}$ due to the one-bit constraint. As the matrices in $\{\pm 1\}^{4 \times 4}$ are purely real, their inverse 2D-DFTs are mirror symmetric \cite{jain1989fundamentals}, i.e., a pattern that illuminates $(k,\ell)$ beamspace index must also illuminates the $(\langle -k \rangle_N, \langle -\ell \rangle_N)$. Due to the mirror symmetry constraint in the illuminations associated with a one-bit phased array, the patterns in Fig.~\ref{fig_2:sectorTypes}\textcolor{red}{(a)} cannot be realized in one-bit phased arrays. In this paper, we design comb-like illumination patterns such as the ones shown in Fig.~\ref{fig_2:sectorTypes}\textcolor{red}{(b)}. Such comb-like patterns exhibit mirror symmetry and the corresponding AWMs can be realized even with low-resolution phased arrays.
\subsection{Proposed construction for AWMs in SLS}\label{subsec_base_awm}
\par Now, we explain the structure of the proposed comb-like patterns and discuss our method to design the base AWMs $\{\mathbf{P}_s\}_{s=0}^{S-1}$ with entries in $\mathbb{Q}^{N \times N}_q$ to achieve such patterns. 
\par We drop the sector index $s$ to explain the first step in our technique to construct AWMs for SLS. To design comb-like structures in the beamspace, our method constructs an antenna domain building block $\mathbf{C}$ with $N_{\mathrm{e}}$ rows and $N_{\mathrm{a}}$ columns. We assume that both $N_{\mathrm{e}}$ and $N_{\mathrm{a}}$ divide $N$, and the product $N_{\mathrm{e}}N_{\mathrm{a}}$ is a power of $2$. The building block $\mathbf{C}$ is chosen from the 2D-DFT codebook such that it illuminates just one direction in the $N_{\mathrm{e}}\times N_{\mathrm{a}}$ beamspace. An example for $\mathbf{C}$ is shown in Fig. \ref{fig_5:sec_design_3steps}\textcolor{red}{(a)}. With our design, the inverse 2D-DFT of $\mathbf{C}$, i.e., $\mathbf{U}_{N_{\mathrm{e}}}^{\ast} \mathbf{C} \mathbf{U}_{N_{\mathrm{a}}}^{\ast}$, has just one non-zero entry and $N_{\mathrm{e}}N_{\mathrm{a}}-1$ zeros. Next, the matrix $\mathbf{C}$ is upsampled by a factor of $N/N_{\mathrm{e}}$ along the row dimension, and by a factor of $N/N_{\mathrm{a}}$ along the column dimension. We define these upsampling factors as 
\begin{align}\label{eqn:rho_e}
\rho_{\mathrm{e}} &= {N}/{N_{\mathrm{e}}} \; \mathrm{and}\\
\label{eqn:rho_a}
\rho_{\mathrm{a}} &= {N}/{N_{\mathrm{a}}},
\end{align}
and the upsampled version of $\mathbf{C}$ is defined as $\tilde{\mathbf{C}} \in \mathbb{C}^{N \times N}$. As upsampling a signal results in replication and scaling in the Fourier domain, it can be shown that the beamspace representation of $\tilde{\mathbf{C}}$ is a scaled repetition of the inverse Fourier transform of the building block \cite{manolakis2011applied}. This repetitive pattern results in a comb-like structure in the beamspace as shown in Fig. \ref{fig_5:sec_design_3steps}. 
\par Now, we summarize the final step in base AWM design to construct a phased array compatible matrix which achieves a comb-like beam pattern. Although the upsampled building block $\tilde{\mathbf{C}}$ results in a comb-like pattern, it cannot be realized with a phased array. This is because  $\tilde{\mathbf{C}} \notin \mathbb{Q}^{N \times N}_q$ as it has several zeros. To construct a matrix in $\mathbb{Q}^{N \times N}_q$ that achieves a comb-like pattern, we first make use of the property that circularly shifting a matrix does not change the magnitude of its 2D-DFT. As a result, any 2D-circular shift of $\tilde{\mathbf{C}}$ results in the same illumination pattern as $\tilde{\mathbf{C}}$. We also observe that the locations of the zeros are complementary across all the 2D-circular shifts of $\tilde{\mathbf{C}}$. Our method takes a weighted combination of the $\rho_{\mathrm{e}}\rho_{\mathrm{a}}$ 2D-circular shifts of $\tilde{\mathbf{C}}$ to construct a base AWM in $\mathbb{Q}^{N \times N}_q$. The weights are optimized such that the illumination pattern associated with this weighted combination is almost flat within the designed sector. We will show in Sec. \ref{sec6_mse} that such flat profiles achieve a low reconstruction error in the sparse estimate. Our design procedure is illustrated in Fig. \ref{fig_5:sec_design_3steps} for a particular sector.       
\par We now explain the mathematical details of our procedure to construct $S = N_\mathrm{a}N_\mathrm{e}$ sectors. As the entries of the building block are drawn from an $N_\mathrm{e} \times N_\mathrm{a}$ 2D-DFT codebook, realizing the designed AWMs in a phased array requires $q=\mathrm{log}_2(\max(\{N_\mathrm{a}, N_\mathrm{e}\}))$-bit phase shifters. We index the $N_\mathrm{a}N_\mathrm{e}$ sectors in our design using the 2D-index pair $(k_\mathrm{e},k_\mathrm{a})$, where $k_\mathrm{e}\in[N_\mathrm{e}]$ and $k_\mathrm{a}\in[N_\mathrm{a}]$. The sector index $s$ is expressed as $s=N_\mathrm{a}k_\mathrm{e} + k_\mathrm{a}$. The set of beamspace indices associated with this sector is
\begin{equation}\label{eqn:comb_sec}
    \mathcal{A}_s = \{(p,q): p = nN_{\mathrm{e}}+k_{\mathrm{e}}, q = mN_{\mathrm{a}}+k_{\mathrm{a}}, n = [\rho_{\mathrm{e}}], m = [\rho_{\mathrm{a}}]\},
\end{equation}
which can be interpreted as a comb-like pattern anchored at $(k_\mathrm{e},k_\mathrm{a})$. A detailed description of our method to construct a base AWM $\mathbf{P}_s$ that illuminates $\mathcal{A}_s$ is given below.
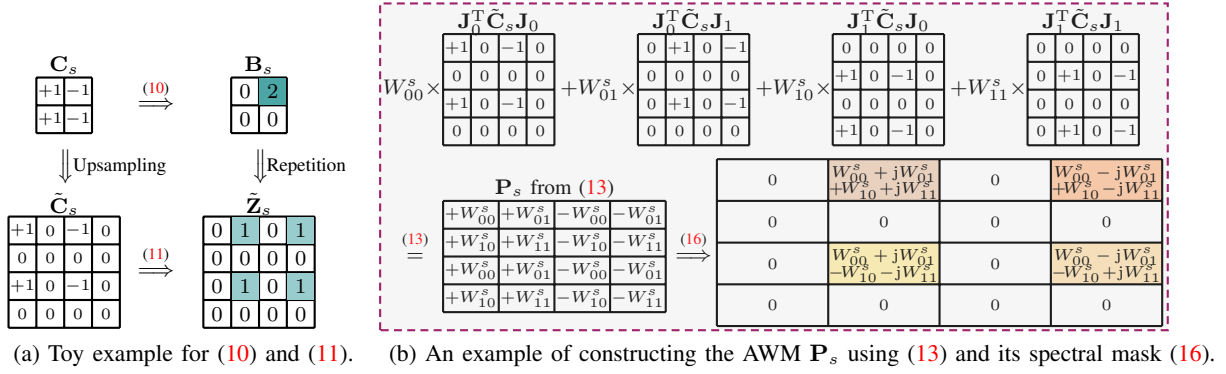
\begin{figure}[t]
\centering
\subfloat[Toy example for \eqref{eqn:block_B1} and \eqref{eqn:z_hat}.]{\begin{tikzpicture}[scale=0.37, every node/.style={minimum size=0.37cm-\pgflinewidth, outer sep=0pt}]
\node at (2, 3.45)   (a) {\scriptsize $\tilde{\mathbf{C}}_s$};
\draw[step=1cm,black,thick] (0,-1) grid (4,3);
\node at (0.5,2.5) {\tiny$+1$};
\node at (1.5,2.5) {\tiny$0$};
\node at (2.5,2.5) {\tiny$-1$};
\node at (3.5,2.5) {\tiny$0$};
\node at (0.5,1.5) {\tiny$0$};
\node at (1.5,1.5) {\tiny$0$};
\node at (2.5,1.5) {\tiny$0$};
\node at (3.5,1.5) {\tiny$0$};
\node at (0.5,0.5) {\tiny$+1$};
\node at (1.5,0.5) {\tiny$0$};
\node at (2.5,0.5) {\tiny$-1$};
\node at (3.5,0.5) {\tiny$0$};
\node at (0.5,-0.5) {\tiny$0$};
\node at (1.5,-0.5) {\tiny$0$};
\node at (2.5,-0.5) {\tiny$0$};
\node at (3.5,-0.5) {\tiny$0$};
\node[rotate=270] at (9,4.8) {\scriptsize{$\Longrightarrow$}};
\node at (10.7,4.8) {\scalebox{0.6}{Repetition}};
\node at (9, 3.45)   (a) {\scriptsize $\tilde{\mathbf{Z}}_s$};
\draw[step=1cm,black,thick] (7,-1) grid (11,3);
\node[fill=teal!40] at (10.5,0.5) {};
\node[fill=teal!40] at (8.5,0.5) {};
\node[fill=teal!40] at (10.5,2.5) {};
\node[fill=teal!40] at (8.5,2.5) {};
\node at (10.5,-0.5) {\scriptsize$0$};
\node at (10.5,0.5) {\scriptsize$1$};
\node at (10.5,1.5) {\scriptsize$0$};
\node at (10.5,2.5) {\scriptsize$1$};
\node at (9.5,-0.5) {\scriptsize$0$};
\node at (9.5,0.5) {\scriptsize$0$};
\node at (9.5,1.5) {\scriptsize$0$};
\node at (9.5,2.5) {\scriptsize$0$};
\node at (8.5,-0.5) {\scriptsize$0$};
\node at (8.5,0.5) {\scriptsize$1$};
\node at (8.5,1.5) {\scriptsize$0$};
\node at (8.5,2.5) {\scriptsize$1$};
\node at (7.5,-0.5) {\scriptsize$0$};
\node at (7.5,0.5) {\scriptsize$0$};
\node at (7.5,1.5) {\scriptsize$0$};
\node at (7.5,2.5) {\scriptsize$0$};
\node at (5.3,1) {\scriptsize{$\Longrightarrow$}};
\node at (5.3,1.6) {\tiny{\eqref{eqn:z_hat}}};
\node[rotate=270] at (2,4.8) {\scriptsize{$\Longrightarrow$}};
\node at (4,4.8) {\scalebox{0.6}{Upsampling}};
\node at (2, 8.45)   (a) {\scriptsize $\mathbf{C}_s$};
\draw[step=1cm,black,thick] (1,6) grid (3,8);
\node at (1.5,7.5) {\tiny$+1$};
\node at (2.5,7.5) {\tiny$-1$};
\node at (1.5,6.5) {\tiny$+1$};
\node at (2.5,6.5) {\tiny$-1$};
\node at (9, 8.45)   (a) {\scriptsize $\mathbf{B}_s$};
\draw[step=1cm,black,thick] (8,6) grid (10,8);
\node[fill=teal!70] at (9.5,7.5) {};
\node at (9.5,7.5) {\scriptsize$2$};
\node at (8.5,7.5) {\scriptsize$0$};
\node at (9.5,6.5) {\scriptsize$0$};
\node at (8.5,6.5) {\scriptsize$0$};
\node at (5.3,7) {\scriptsize{$\Longrightarrow$}};
\node at (5.3,7.6) {\tiny{\eqref{eqn:block_B1}}};
\end{tikzpicture}
}
\hfil
\subfloat[An example of constructing the AWM $\mathbf{P}_s$ using \eqref{eqn:AWM_from_upsampled} and its spectral mask \eqref{eqn:mask_from_circ}.]{\begin{tikzpicture}[scale=0.37, every node/.style={minimum size=0.37cm-\pgflinewidth, outer sep=0pt}]
\node at (2, 9.45)   (a) {\scriptsize $\mathbf{J}_{0}^{\mathrm{T}}\tilde{\mathbf{C}}_s\mathbf{J}_{0}$};
\draw[step=1cm,black,thick] (0,5) grid (4,9);
\node at (-1.1,7) {\scriptsize $W^s_{00}\times$};
\node at (0.5,8.5) {\tiny$+1$};
\node at (0.5,7.5) {\tiny$0$};
\node at (0.5,6.5) {\tiny$+1$};
\node at (0.5,5.5) {\tiny$0$};
\node at (1.5,8.5) {\tiny$0$};
\node at (1.5,7.5) {\tiny$0$};
\node at (1.5,6.5) {\tiny$0$};
\node at (1.5,5.5) {\tiny$0$};
\node at (2.5,8.5) {\tiny$-1$};
\node at (2.5,7.5) {\tiny$0$};
\node at (2.5,6.5) {\tiny$-1$};
\node at (2.5,5.5) {\tiny$0$};
\node at (3.5,8.5) {\tiny$0$};
\node at (3.5,7.5) {\tiny$0$};
\node at (3.5,6.5) {\tiny$0$};
\node at (3.5,5.5) {\tiny$0$};
\node at (9, 9.45)   (a) {\scriptsize $\mathbf{J}_{0}^{\mathrm{T}}\tilde{\mathbf{C}}_s\mathbf{J}_{1}$};
\draw[step=1cm,black,thick] (7,5) grid (11,9);
\node at (5.6,7) {\scriptsize $+W^s_{01}\times$};
\node at (7.5,8.5) {\tiny$0$};
\node at (7.5,7.5) {\tiny$0$};
\node at (7.5,6.5) {\tiny$0$};
\node at (7.5,5.5) {\tiny$0$};
\node at (8.5,8.5) {\tiny$+1$};
\node at (8.5,7.5) {\tiny$0$};
\node at (8.5,6.5) {\tiny$+1$};
\node at (8.5,5.5) {\tiny$0$};
\node at (9.5,8.5) {\tiny$0$};
\node at (9.5,7.5) {\tiny$0$};
\node at (9.5,6.5) {\tiny$0$};
\node at (9.5,5.5) {\tiny$0$};
\node at (10.5,8.5) {\tiny$-1$};
\node at (10.5,7.5) {\tiny$0$};
\node at (10.5,6.5) {\tiny$-1$};
\node at (10.5,5.5) {\tiny$0$};
\node at (16, 9.45)   (a) {\scriptsize $\mathbf{J}_{1}^{\mathrm{T}}\tilde{\mathbf{C}}_s\mathbf{J}_{0}$};
\draw[step=1cm,black,thick] (14,5) grid (18,9);
\node at (12.6,7) {\scriptsize $+W^s_{10}\times$};
\node at (14.5,8.5) {\tiny$0$};
\node at (14.5,7.5) {\tiny$+1$};
\node at (14.5,6.5) {\tiny$0$};
\node at (14.5,5.5) {\tiny$+1$};
\node at (15.5,8.5) {\tiny$0$};
\node at (15.5,7.5) {\tiny$0$};
\node at (15.5,6.5) {\tiny$0$};
\node at (15.5,5.5) {\tiny$0$};
\node at (16.5,8.5) {\tiny$0$};
\node at (16.5,7.5) {\tiny$-1$};
\node at (16.5,6.5) {\tiny$0$};
\node at (16.5,5.5) {\tiny$-1$};
\node at (17.5,8.5) {\tiny$0$};
\node at (17.5,7.5) {\tiny$0$};
\node at (17.5,6.5) {\tiny$0$};
\node at (17.5,5.5) {\tiny$0$};
\node at (23, 9.45)   (a) {\scriptsize $\mathbf{J}_{1}^{\mathrm{T}}\tilde{\mathbf{C}}_s\mathbf{J}_{1}$};
\draw[step=1cm,black,thick] (21,5) grid (25,9);
\node at (19.6,7) {\scriptsize $+W^s_{11}\times$};
\node at (21.5,8.5) {\tiny$0$};
\node at (21.5,7.5) {\tiny$0$};
\node at (21.5,6.5) {\tiny$0$};
\node at (21.5,5.5) {\tiny$0$};
\node at (22.5,8.5) {\tiny$0$};
\node at (22.5,7.5) {\tiny$+1$};
\node at (22.5,6.5) {\tiny$0$};
\node at (22.5,5.5) {\tiny$+1$};
\node at (23.5,8.5) {\tiny$0$};
\node at (23.5,7.5) {\tiny$0$};
\node at (23.5,6.5) {\tiny$0$};
\node at (23.5,5.5) {\tiny$0$};
\node at (24.5,8.5) {\tiny$0$};
\node at (24.5,7.5) {\tiny$-1$};
\node at (24.5,6.5) {\tiny$0$};
\node at (24.5,5.5) {\tiny$-1$};
\node at (4, 3.45)   (a) {\scriptsize $\mathbf{P}_s$ from \eqref{eqn:AWM_from_upsampled}};
\draw[xstep=2cm,ystep=1,black,thick] (0,-1) grid (8,3);
\node at (-1,1) {\scriptsize{$=$}};
\node at (-1,1.6) {\tiny{\eqref{eqn:AWM_from_upsampled}}};
\node at (1,2.5) {\tiny$+W^s_{00}$};
\node at (1,1.5) {\tiny$+W^s_{10}$};
\node at (1,0.5) {\tiny$+W^s_{00}$};
\node at (1,-0.5) {\tiny$+W^s_{10}$};
\node at (3,2.5) {\tiny$+W^s_{01}$};
\node at (3,1.5) {\tiny$+W^s_{11}$};
\node at (3,0.5) {\tiny$+W^s_{01}$};
\node at (3,-0.5) {\tiny$+W^s_{11}$};
\node at (5,2.5) {\tiny$-W^s_{00}$};
\node at (5,1.5) {\tiny$-W^s_{10}$};
\node at (5,0.5) {\tiny$-W^s_{00}$};
\node at (5,-0.5) {\tiny$-W^s_{10}$};
\node at (7,2.5) {\tiny$-W^s_{01}$};
\node at (7,1.5) {\tiny$-W^s_{11}$};
\node at (7,0.5) {\tiny$-W^s_{01}$};
\node at (7,-0.5) {\tiny$-W^s_{11}$};
\node at (9,1) {\scriptsize{$\Longrightarrow$}};
\node at (9,1.6) {\tiny{\eqref{eqn:mask_from_circ}}};
\hspace{-0.8cm}\draw [fill=Tan!40,draw=none] (16,3) rectangle (20,4.5);
\draw [fill=Goldenrod!45,draw=none] (16,0) rectangle (20,1.5);
\draw [fill=Orange!40,draw=none] (24,3) rectangle (28,4.5);
\draw [fill=Dandelion!35,draw=none] (24,0) rectangle (28,1.5);
\draw[xstep=4cm,ystep=1.5,black,thick] (12,-1.5) grid (28,4.5);
\node at (14,-0.75) {\tiny$0$};
\node at (14,0.75) {\tiny$0$};
\node at (14,2.25) {\tiny$0$};
\node at (14,3.75) {\tiny$0$};
\node at (18,4) {\tiny$W^s_{00}+\mathrm{j}W^s_{01}$};
\node at (18,3.4) {\tiny$+\!W^s_{10}\!+\!\mathrm{j}W^s_{11}$};
\node at (18,-0.75) {\tiny$0$};
\node at (18,1) {\tiny$W^s_{00}+\mathrm{j}W^s_{01}$};
\node at (18,0.4) {\tiny$-\!W^s_{10}\!-\!\mathrm{j}W^s_{11}$};
\node at (18,2.25) {\tiny$0$};
\node at (22,-0.75) {\tiny$0$};
\node at (22,0.75) {\tiny$0$};
\node at (22,2.25) {\tiny$0$};
\node at (22,3.75) {\tiny$0$};
\node at (26,-0.75) {\tiny$0$};
\node at (26,4) {\tiny$W^s_{00}-\mathrm{j}W^s_{01}$};
\node at (26,3.4) {\tiny$+\!W^s_{10}\!-\!\mathrm{j}W^s_{11}$};
\node at (26,2.25) {\tiny$0$};
\node at (26,1) {\tiny$W^s_{00}-\mathrm{j}W^s_{01}$};
\node at (26,0.4) {\tiny$-\!W^s_{10}\!+\!\mathrm{j}W^s_{11}$};
\draw[RedViolet,densely dashed,fill=Gray!50,fill opacity=0.15,thick] (-0.1, -1.65) rectangle (28.2, 10.1) {};
\end{tikzpicture}}
\caption{\small An example of our proposed framework to construct AWMs that focus energy within sectors with comb-like patterns. Here, $N=4$, the number of sectors is $S=4$, $N_{\mathrm{e}}=2$,  and $N_{\mathrm{a}}=2$. Therefore, $\rho_\mathrm{e}=2$ and $\rho_\mathrm{a}=2$.\normalsize}\label{fig_5:sec_design_3steps} 
\end{figure}
\begin{itemize}
     \item Construct the building block $\mathbf{C}_s\in\mathbb{Q}_q^{N_{\mathrm{e}}\times N_{\mathrm{a}}}$ as an outer product of the $k_{\mathrm{e}}^{\text{th}}$ column of $\mathbf{U}_{N_\mathrm{e}}$ and the $k_{\mathrm{a}}^{\text{th}}$ column of $\mathbf{U}_{N_\mathrm{a}}$, i.e., 
     \begin{equation}
         \mathbf{C}_s = \mathbf{U}_{N_\mathrm{e}}(:,k_\mathrm{e})\mathbf{U}_{N_\mathrm{a}}^{T}(:,k_\mathrm{a}).
     \end{equation}
     Similar to \cite{myers2019falp}, we define $\mathbf{C}_{s,\mathrm{FC}}$ as the flipped and conjugated version of $\mathbf{C}_s$. The angle domain matrix associated with the building block is then
    \begin{equation}\label{eqn:block_B1}
     \mathbf{B}_s = \mathbf{U}_{N_{\mathrm{e}}}^{*}\mathbf{C}_{s,\mathrm{FC}}\mathbf{U}_{N_{\mathrm{a}}}^{*}
    \end{equation}
     It can be shown that $B_s(k_{\mathrm{e}},k_{\mathrm{a}}) = 1$ and $B_s(i,j) = 0~ \forall (i,j)\neq (k_{\mathrm{e}},k_{\mathrm{a}})$.
 \item Upsample $\mathbf{C}_s$ by a factor of $\rho_\mathrm{e}$ along the columns and $\rho_\mathrm{a}$ along the rows. The upsampled matrix is obtained by first inserting $\rho_\mathrm{e}-1$ zero-valued row vectors between successive rows of $\mathbf{C}_s$, and then inserting $\rho_\mathrm{a}-1$ zero-valued column vectors between successive columns of the resultant matrix. The upsampled result is an $N\times N$ matrix $\tilde{\mathbf{C}}_s$, whose beamspace representation is defined as 
 \begin{align}
 \tilde{\mathbf{Z}}_s &= \mathbf{U}_{N}^{*}\tilde{\mathbf{C}}_{s,\mathrm{FC}}\mathbf{U}_{N}^{*}\label{eqn:z_hat}\\[0.8em]
      &= \frac{1}{\sqrt{\rho_{\mathrm{e}}\rho_{\mathrm{a}}}}\vphantom{
    \begin{matrix}
 \overbrace{}^{\mbox{$\rho_{\mathrm{a}}~\text{times}$}}
    \end{matrix}}%
\begin{bmatrix}
\coolover{\rho_{\mathrm{a}}~\text{times}}{\mathbf{B}_s & \cdots & \mathbf{B}_s}\\
    \vdots & \ddots & \vdots\\
    \mathbf{B}_s & \cdots & \mathbf{B}_s
\end{bmatrix}%
\hspace{-0.2cm}
\begin{matrix}
    \coolrightbrace{x \\ \vspace{-0.2cm}x\\ y}{\rho_{\mathrm{e}}~\text{times}}
\end{matrix}.
\label{eqn:zhat_explicit}
\end{align}
Here, \eqref{eqn:z_hat} follows from the upsampling property of the 2D-DFT \cite{kak2001principles}. Observe that $\tilde{\mathbf{Z}}_s$ exhibits a comb-like structure as it contains repetitions of $\mathbf{B}_s$, that has a single non-zero entry. The pattern in $\tilde{\mathbf{Z}}_s$, however, cannot be realized in a phased array as its antenna domain representation $\tilde{\mathbf{C}}_s \notin \mathbb{Q}_q^{N \times N}$.
\item Express $\mathbf{P}_s$ as a weighted sum of $\rho_\mathrm{e}\rho_\mathrm{a}$ distinct 2D-circular shifts of $\tilde{\mathbf{C}}_s$. Let $\mathbf{J}\in\mathbb{R}^{N\times N}$ denote a circulant delay matrix with the first row of $(0,1,0,...,0)$. The subsequent rows of $\mathbf{J}$ are generated by circularly shifting the previous row by 1 unit. The $d$ circulant delay matrix is then $\mathbf{J}_d = \mathbf{J}\cdot\mathbf{J}\cdots\mathbf{J}$. The base AWM $\mathbf{P}_s$ in our construction is then
\begin{equation}\label{eqn:AWM_from_upsampled}
    \mathbf{P}_s = \sum\limits_{\ell=0}^{\rho_{\mathrm{e}}-1}\sum\limits_{m=0}^{\rho_{\mathrm{a}}-1}W^s_{\ell m}\mathbf{J}_{\ell}^{T}\tilde{\mathbf{C}}_s\mathbf{J}_{m}.
\end{equation}
We constraint the entries of the $\rho_{\mathrm{e}}\times \rho_{\mathrm{a}}$ weight matrix $\mathbf{W}^s$ in \eqref{eqn:AWM_from_upsampled} to $\mathbb{Q}_{q}$  so that $\mathbf{P}_s \in \mathbb{Q}_{q}^{N\times N}$. 
\end{itemize}
\par  A possible choice for $\mathbf{W}^s$ is to set its entries to the elements in $\mathbb{Q}_q$ at random. Although a random choice still illuminates the sector of interest, it does not necessarily lead to a \textit{uniform} illumination pattern across the directions within the sector. Specifically, the entries of $\mathbf{Z}_s$ within the set $\mathcal{A}_s$ will not necessarily have the same magnitude as we can also notice from the examples in Fig.~\ref{fig_5:sec_design_3steps}\textcolor{red}{(b)} and Fig.~\ref{fig_6:SpectralMaskWeights}\textcolor{red}{(b)}. In this paper, the weights in $\mathbf{W}^s$ are optimized such that the illumination, i.e., the transmitted power, is almost the same across all the directions within each sector. We will show in Section \ref{sec6_mse} that having a flat illumination pattern within a sector results in a tight bound on the reconstruction error of the in-sector channel via CS. 
\par We explain how the weights in $\mathbf{W}^s$ are optimized to achieve an almost flat illumination pattern within the sector $s$. We define an $N\times N$ matrix 
\begin{equation}
\label{eqn:Ts_definition}
\mathbf{T}^{s} = \sum\limits_{\ell=0}^{\rho_{\mathrm{e}}-1}\sum\limits_{m=0}^{\rho_{\mathrm{a}}-1}W^s_{\ell m}\mathbf{a}_{N}\left(2\pi \ell/N\right) \mathbf{a}_{N}^\mathrm{T}\left(2\pi m/N\right),
\end{equation}
where $\mathbf{a}_{N}(\omega)$ is the $N\times 1$ Vandermonde vector defined in \eqref{eqn:vandrmond}. The beamspace representation of the base AWM $\mathbf{P}_s$ in \eqref{eqn:AWM_from_upsampled}, referred to as the spectral mask $\mathbf{Z}_s$, is given by
\begin{align}
    \mathbf{Z}_s &= N\mathbf{U}_{N}^{*}\mathbf{P}_{s,\mathrm{FC}}\mathbf{U}_{N}^{*}\label{eqn:mask_from_circ_1}\\
    &= N[\mathbf{T}^s]^{\mathrm{c}}\odot\tilde{\mathbf{Z}}_s\label{eqn:mask_from_circ}.
\end{align}
We arrive at \eqref{eqn:mask_from_circ} by first applying flipping and conjugation operation to both sides of \eqref{eqn:AWM_from_upsampled}. Then, we use the property that circularly shifting a matrix is equivalent to modulation in the Fourier domain. 
Next, we observe from \eqref{eqn:zhat_explicit} that $\tilde{\mathbf{Z}}_s$ has equal magnitude entries within the sector $\mathcal{A}_s$ by our construction. For the base AWM $\mathbf{P}_s$ to achieve a uniform illumination over $\mathcal{A}_s$, we notice from \eqref{eqn:mask_from_circ} that the entries of $\mathbf{T}^{s}$ must have the same magnitude over $\mathcal{A}_s$. We define $\mathbf{T}_{\mathcal{A}s}^s$ as the $\rho_{\mathrm{e}}\times\rho_{\mathrm{a}}$ submatrix of $\mathbf{T}^s$, comprising entries from $\mathbf{T}^s$ at the indices in $\mathcal{A}_s$. Further, we define an $N\times N$ diagonal matrix $\mathbf{D}_{N}(\omega) = \mathrm{Diag}(\mathbf{a}_{N}(\omega))$ that contains the $N\times 1$ Vandermonde vector \eqref{eqn:vandrmond} on its diagonal. Then, we can rewrite \eqref{eqn:Ts_definition} as 
\begin{equation}\label{eqn:weighted_w}
    \mathbf{T}_{\mathcal{A}_s}^s = \mathbf{U}_{\rho_\mathrm{e}}^{*} \mathbf{D}_{\rho_\mathrm{e}}(2\pi k_{\mathrm{e}}/N)\mathbf{W}^{s}\mathbf{D}_{\rho_\mathrm{a}}(2\pi k_{\mathrm{a}}/N)\mathbf{U}_{\rho_\mathrm{a}}^{*}.
\end{equation}
The weight matrix $\mathbf{W}^{s}\in\mathbb{Q}_q^{\rho_{\mathrm{e}}\times\rho_{\mathrm{a}}}$ should be optimized such that $\mathbf{T}_{\mathcal{A}_s}^s$, i.e., the magnitude of the inverse DFT of $\mathbf{D}_{\rho_\mathrm{e}}(2\pi k_{\mathrm{e}}/N)\mathbf{W}^{s}\mathbf{D}_{\rho_\mathrm{a}}(2\pi k_{\mathrm{a}}/N)$ is flat. We can thus write the optimization problem for designing $\mathbf{W}^s$ as
\begin{equation}\label{eqn:opt_weight}
\mathcal{O}:
\underset{\mathbf{W}^{s}\in\mathbb{Q}_q^{\rho_{\mathrm{e}}\times \rho_{\mathrm{a}}},|\mathbf{V}|=\mathbf{1}}{{\mathrm{minimize}}} \ \ \ \|\mathbf{T}_{\mathcal{A}_s}^s-\mathbf{V}/\sqrt{\rho_{\mathrm{e}}\rho_{\mathrm{a}}}\|_{\mathrm{F}}.
\end{equation}
The above optimization problem makes sure that the magnitude of entries of $\mathbf{T}_{\mathcal{A}_s}^s$ are as close as possible to $1/\sqrt{\rho_{\mathrm{e}}\rho_{\mathrm{a}}}$ while $\mathbf{W}^s\in\mathbb{Q}_q^{\rho_{\mathrm{e}}\times \rho_{\mathrm{a}}}$.
\begin{figure}[t]
\centering  
\subfloat[Illumination pattern $|\mathbf{Z}_s|$ with the proposed construction and the use of optimized weights for $\mathbf{W}^s$.]{\includegraphics[trim=4.24cm 7cm 1.75cm 7cm, clip, scale=0.31]{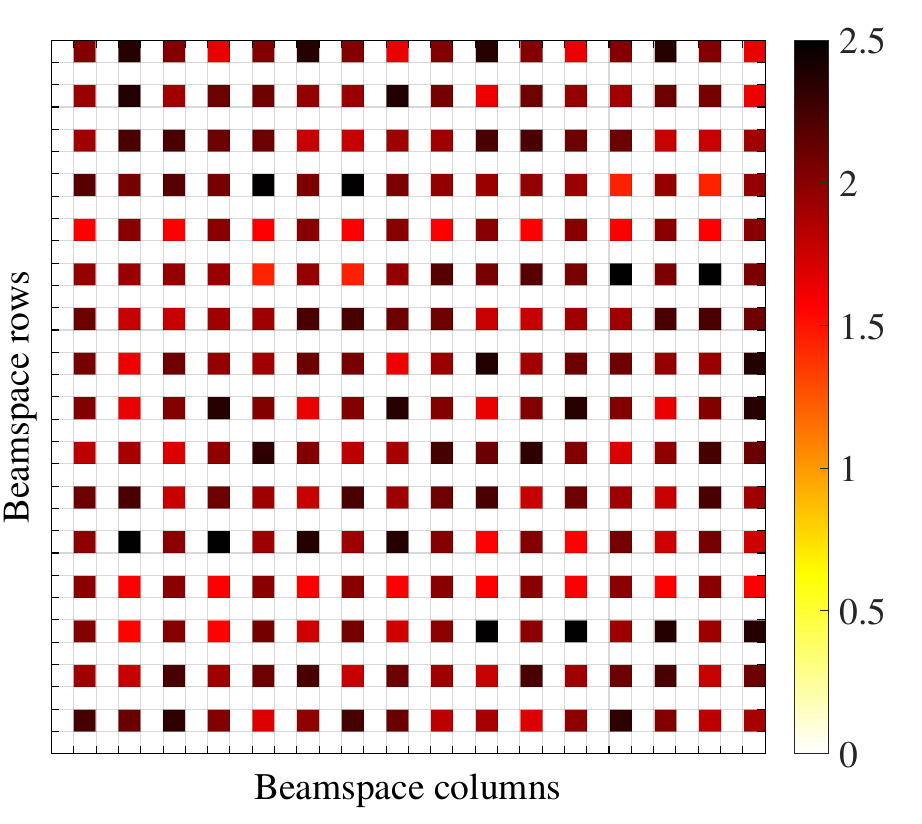}}
\hfil
\subfloat[Illumination pattern $|\mathbf{Z}_s|$ with the proposed construction and the use of random weights for $\mathbf{W}^s$.]{\includegraphics[trim=4.24cm 7cm 1.75cm 7cm, clip, scale=0.31]{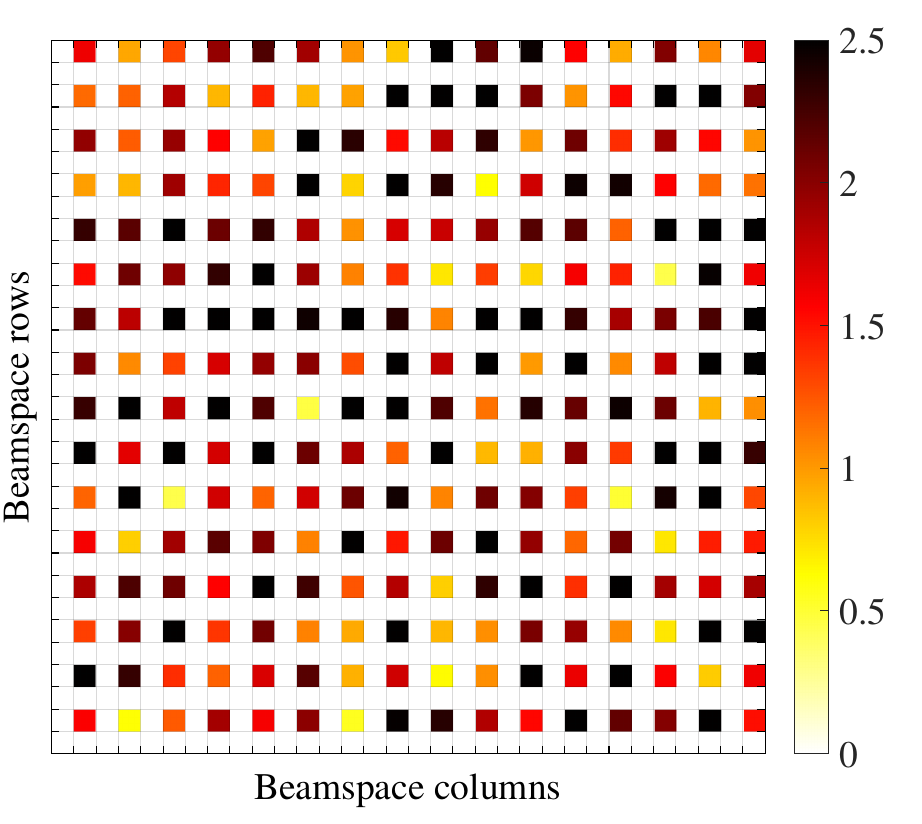}
}
\caption{\small This figure illustrates the spectral masks of our base AWMs designed for SLS. In this example, we use $N=32$, $N_{\mathrm{e}}=2$, $N_{\mathrm{a}}=2$, $k_{\mathrm{e}}=0$, $k_{\mathrm{a}}=1$, and $q=2$ bits. We observe from Fig. (a) and Fig. (b) that using the optimized $\mathbf{W}^s$ in our comb-like construction results in a more uniform illumination pattern than random weights. Here, $\underset{(i,j)\in\mathcal{A}_s}{{\max}}|Z_{s,ij}|/\underset{(i,j)\in\mathcal{A}_s}{{\min}}|Z_{s,ij}|$ is about $1.78$ for the pattern in Fig. (a) and is $8.67$ for the pattern in Fig. (b).\normalsize}\label{fig_6:SpectralMaskWeights} 
\end{figure}
\par In our work, we develop a simple adaptation of the PeCAN algorithm proposed in \cite{stoica2009designing} to optimize $\mathbf{W}^s$ in \eqref{eqn:opt_weight}. The PeCAN algorithm is an iterative approach to design signals with a flat spectral magnitude. Our adaption incorporates the diagonal matrices $\mathbf{D}_{\rho_\mathrm{e}}(2\pi k_{\mathrm{e}}/N)$ and $\mathbf{D}_{\rho_\mathrm{a}}(2\pi k_{\mathrm{a}}/N)$ into the objective of the PeCAN algorithm. Further, the weights at the end of each iteration are quantized so that $\mathbf{W}^{s}\in\mathbb{Q}_q^{\rho_{\mathrm{e}} \times \rho_{\mathrm{a}}}$. As seen in Fig.~\ref{fig_6:SpectralMaskWeights}, the magnitude profile of the spectral mask $|\mathbf{Z}_s|$ corresponding to the optimized weights has a more uniform illumination within the sector than the one that uses random weights. 
Finally, each sector's weight matrix $\mathbf{W}^s$ is designed independent of the other sectors. The design is optimized only once using the modified PeCAN algorithm and the optimized AWM for the SLS codebook can be stored.
\begin{figure}[t]
\centering
\subfloat[$S\!=\!2$ sectors and $q\!=\!1$ bit.]{\includegraphics[trim=1.5cm 6cm 2cm 7.5cm, clip, scale=0.3]{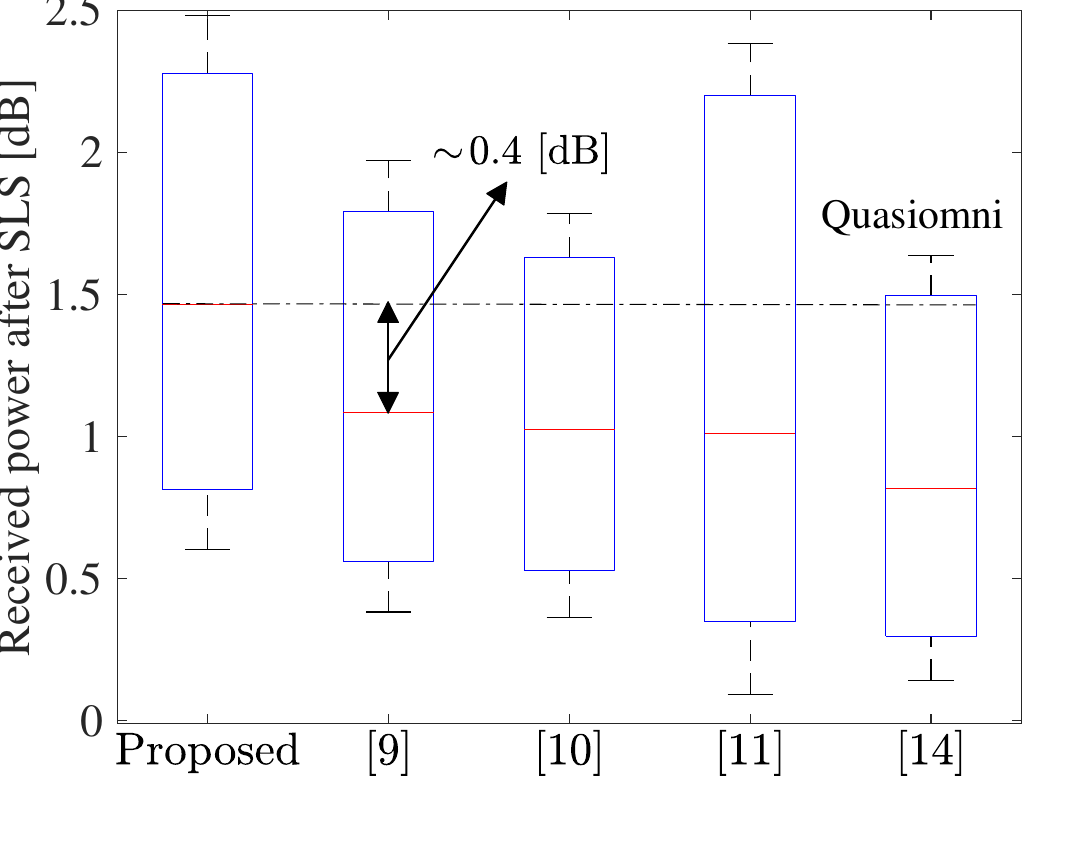}}
\hfil
\subfloat[$S\!=\!4$ sectors and $q\!=\!1$ bit.]{\includegraphics[trim=1.75cm 6cm 2cm 7.5cm, clip, scale=0.3]{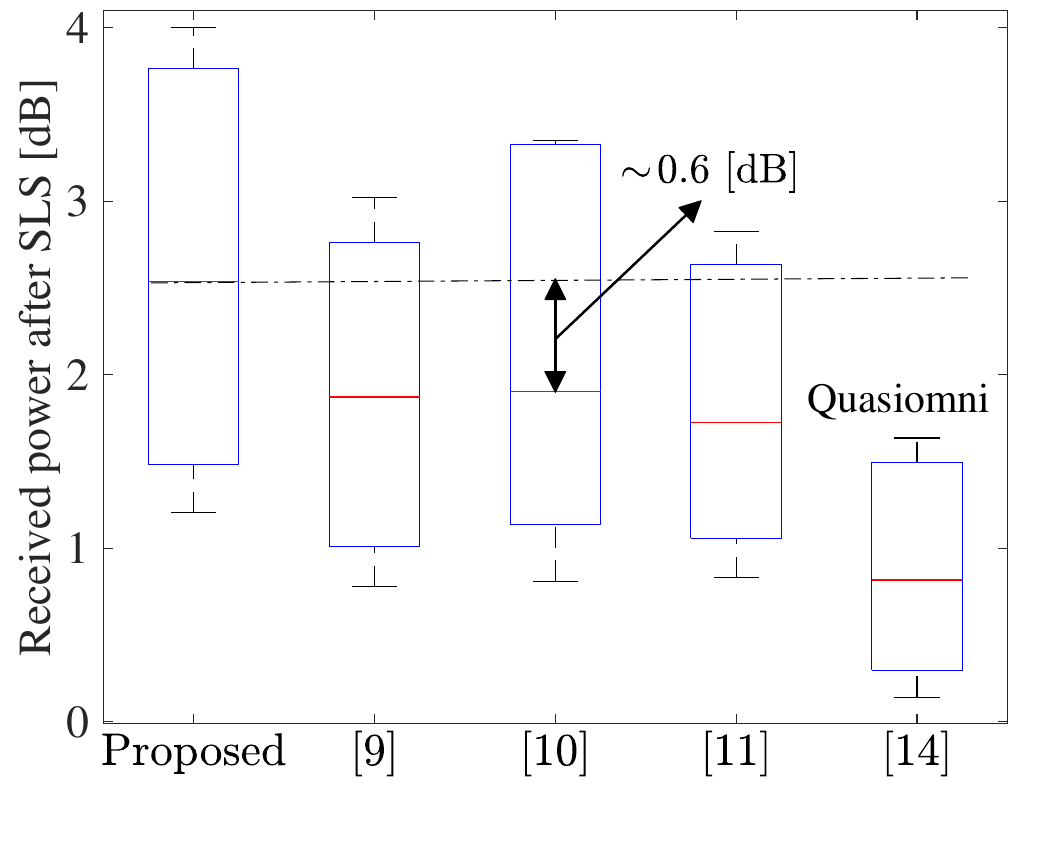}
}
\hfil
\subfloat[$S\!=\!16$ sectors and $q\!=\!2$ bits.]{\includegraphics[trim=1.55cm 6cm 2cm 7.5cm, clip, scale=0.3]{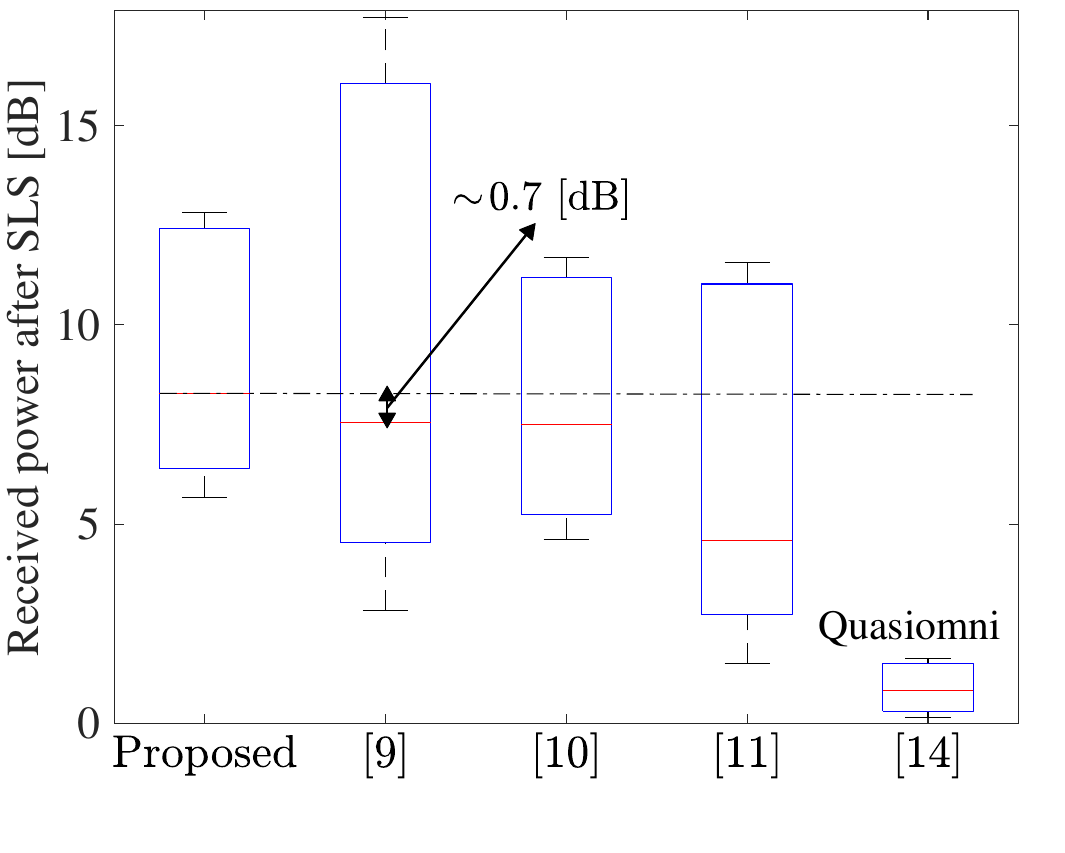}
}
\caption{\small Our designed base AWMs achieve a higher median received power after SLS than the ones in 
\cite{tsai2018structured,xiao2018enhanced,raviteja2017analog} and the quasi-omnidirectional beams from \cite{myers2019falp}. These plots are obtained for a $32\times 32$ array at the TX. The received power is computed for $100$ channel realizations from the NYU channel simulator at $60$ GHz and $\sigma^2=0$. \normalsize}\label{fig:boxplot} 
\end{figure}
\par In Fig.~\ref{fig:boxplot}, we compare the received signal strength after SLS using the proposed comb-like sector codebook with those in \cite{tsai2018structured,xiao2018enhanced,raviteja2017analog}. The AWMs in \cite{tsai2018structured,xiao2018enhanced,raviteja2017analog} have contiguous beam patterns and to meet the low-resolution constraint of the phase shifters, we have quantized the corresponding AWMs when necessary. We note that these methods may select a different sector as the best sector, depending on the beam patterns used in SLS. The received power after SLS is proportional to the SNR in the best estimated sector. We observe from Fig.~\ref{fig:boxplot} that our proposed AWMs achieve a higher received power after SLS than the other methods. This is because our comb-like sectors are able to concentrate the transmitter's energy perfectly within the sectors without any leakage, unlike the other methods. After the sector of interest is found using SLS, the TX performs in-sector CS-based channel estimation within this sector for beam alignment.
\section{In-sector Compressed Sensing}
\par To achieve the full beamforming gain of $10\, \mathrm{log}_{10} N^2$, the TX needs to further refine the AWM within the sector of interest $\mathcal{A}_0$. This beam refinement can be performed after estimating the $N^2/S$ dimension in-sector channel. In this section, we explain how 2D-CCS can be used to acquire channel measurements within the sector of interest. Then, we discuss how to construct an ensemble of AWMs $\{\Po[m]\}_{m=0}^{M-1}$ that all focus on the sector interest $\Ao$. Finally, we provide details on the subset of AWMs optimized to reduce aliasing artifacts within the sector of interest.

\subsection{Preliminaries on 2D-CCS}\label{sec_InCS}
At the end of SLS, the TX can apply $\Po$ to illuminate the beamspace indices in $\Ao$, i.e., 
\begin{equation}
\label{eq:defn_sector}
    \Ao=\{(p,q): (\mathbf{U}^{\ast}_N \Po \mathbf{U}^{\ast}_N )_{pq} \neq 0\}.
\end{equation}
The use of $\Po$ at the TX, however, results in a single spatial channel measurement, which is insufficient to estimate the $N^2/S$ entries of the beamspace in $\Ao$. To estimate these entries, multiple AWMs that all focus on the same sector must be designed and applied at the TX. One way to construct such AWMs is to circularly shift $\Po$ and obtain possibly distinct AWMs. This approach works for in-sector CS because circularly shifting a matrix does not change the magnitude of its 2D-DFT \cite{jain1989fundamentals}. The resulting method in which the TX applies circular shifts of an AWM for the RX to acquire channel measurements is called 2D-CCS \cite{myers2019falp}. Hence, with 2D-CCS, the $M$ in-sector channel measurements $\{y_\mathrm{o}[m]\}_{m=0}^{M-1}$ are obtained by applying $M$ distinct circular shifts of $\Po$ at the TX during BRP. 
\par Now, we discuss mathematical preliminaries on 2D-CCS \cite{li2012convolutional}. With 2D-CCS, the TX can possibly apply $N^2$ different 2D-circular shifts of $\Po$ for in-sector channel estimation. As applying all possible circular shifts results in a substantial measurement overhead, the TX only applies $M<N^2$ of these circular shifts to obtain channel measurements. These measurements are used together with a sparse prior for channel estimation. As the AWMs in 2D-CCS are all circular shifts of the base AWM $\Po$, the channel measurements can be interpreted as a subsampled circular convolution of the channel and $\Po$ \cite{myers2019falp}. We use $\Omega = \{(r[m],c[m])\}_{m=0}^{M-1}$ to denote the set of $M$ circular shifts of $\Po$ applied at the TX to acquire the in-sector channel measurements. Specifically, $\Po[m]$ is constructed by circularly shifting $\Po$ by $r[m]\in[N]$ units along its rows, and then circularly shifting the result by $c[m]\in[N]$ units along its columns. We use $\mathbf{H}\circledast\Po$ to denote the 2D-circular convolution \cite{kak2001principles} of $\mathbf{H}$ and $\Po$. We define $\mathcal{P}_{\Omega}(\mathbf{A})$ as the subsampling operation that returns a vector of size $|\Omega|\times 1$ containing the entries of $\mathbf{A}$ at locations in $\Omega$. When the TX applies circulant shifts of $\Po$ according to $\Omega$, the vector of $M$ in-sector channel measurements acquired by the RX is
\begin{equation}\label{eqn:2D_CCS_measure}
    \mathbf{y} = \mathcal{P}_{\Omega}(\mathbf{H}\circledast\Po) + \mathbf{v}.
\end{equation}
We observe that the channel measurements in 2D-CCS are determined by the base AWM $\Po$ and the set of circular shifts $\Omega$.
\par We now describe the in-sector sensing structure in \eqref{eqn:2D_CCS_measure} by expressing $\mathbf{H}$ and $\Po$ in the 2D-DFT domain. Similar to \eqref{eqn:mask_from_circ_1}, the spectral mask $\Zo$ corresponding to $\Po$ is defined as 
\begin{equation}\label{eqn:mask}
    \Zo = N\mathbf{U}_{N}^{*}\mathbf{P}_{\mathrm{o},\mathrm{FC}}\mathbf{U}_{N}^{*}.
\end{equation}
From \eqref{eq:defn_sector} and the properties of the 2D-DFT, it can be shown that $Z_\mathrm{o}(k,\ell) \neq 0$ if $(k, \ell) \in \Ao$ and  $Z_\mathrm{o}(k,\ell) = 0$ if $(k, \ell) \notin \Ao$. We define the masked beamspace $\mathbf{R}_\mathrm{o}$ as the entry-wise product of the spectral mask $\Zo$ and the beamspace $\mathbf{X}$, i.e.,
\begin{equation}\label{eqn:msked_beam_R}
    \mathbf{R}_\mathrm{o} = \mathbf{X}\odot\Zo.
\end{equation}
As $\Zo$ contains only $|\Ao|$ number of non-zeros, the masked beamspace $\mathbf{R}_\mathrm{o}$ in \eqref{eqn:msked_beam_R} has utmost $|\Ao|$ non-zeros. Under the assumption that $|\Ao|\sim O(N^2)$, $\mathbf{R}_\mathrm{o}$ exhibits a sparse structure due to the sparsity in $\mathbf{X}$. Now, we use the property that the 2D-DFT of the 2D circular convolution of two matrices is equal to the scaled entry-wise product of their 2D-DFTs \cite{kak2001principles} to rewrite \eqref{eqn:2D_CCS_measure} as
\begin{equation}\label{eqn:2D_CS_measure}
    \mathbf{y} = \mathcal{P}_{\Omega}(\mathbf{U}_{N}\mathbf{R}_\mathrm{o}\mathbf{U}_{N}) + \mathbf{v}.
\end{equation}
We observe from \eqref{eqn:2D_CS_measure} that the channel measurements in 2D-CCS are simply a subsampled 2D-DFT of a sparse masked beamspace. Here, the mask is the discrete beam pattern corresponding to the base AWM $\Po$, and the subsampling locations in $\Omega$ are the circular shifts applied at TX.
\par Now, we explain how in-sector channel estimation is performed with the 2D-CCS-based measurements. Any algorithm that uses the measurements $\mathbf{y} = \mathcal{P}_{\Omega}\left(\mathbf{U}_{N}(\mathbf{X} \odot \Zo)\mathbf{U}_{N}\right)$ in \eqref{eqn:2D_CCS_measure} can only estimate the entries of $\mathbf{X}$ at the indices in $\Ao$, since the spectral mask $\Zo$ blanks out entries of $\mathbf{X}$ outside $\Ao$. In this paper, we use the OMP, a greedy CS algorithm, to obtain the beamspace estimate $\hat{\mathbf{X}}_\mathrm{o}$ from the $M$ in-sector measurements. The channel estimate within the sector of interest is then $\hat{\mathbf{H}}_\mathrm{o}=\mathbf{U}_N \hat{\mathbf{X}}_\mathrm{o} \mathbf{U}_N$.
\par The success of sparse recovery from partial 2D-DFT measurements in \eqref{eqn:2D_CS_measure} depends on the choice of the subsampling locations in $\Omega$. Prior work has demonstrated that subsampling sets chosen at random allow sparse recovery with a high probability \cite{candes2007sparsity}. It is also known that random subsampling often results in aliasing artifacts that are almost uniformly spread over the support of the reconstruction, i.e., $[N] \times [N]$ \cite{lustig2007sparse}. Many CS algorithms, like the OMP, iteratively cancel the aliasing artifacts to reconstruct the sparse signal. In our in-sector CS problem, a uniformly spread aliasing profile is undesirable because the sparse masked beamspace signal in \eqref{eqn:2D_CS_measure} is non-zero only at the indices in $\Ao$. To minimize the reconstruction error, the subsampling set $\Omega$ should be constructed such that aliasing artifacts over $\Ao$ are kept small, while the artifacts outside $\Ao$ are not critical since the masked beamspace is known to be zero in those regions. We construct such sets in Sec. \ref{sec5_CircShift}.

\subsection{Proposed circular shifts for in-sector CS}\label{sec5_CircShift}
\par We discuss our procedure to optimize the set of the circular shifts in $\Omega$ to reduce the aliasing artifacts within the sector of interest $\Ao$. 
\par To aid our optimization, we write down the partial 2D-DFT-based measurement model, i.e., 
\begin{equation}
\label{eqn:partial2DDFT}
\mathbf{y} = \mathcal{P}_{\Omega}(\mathbf{U}_{N}\left(\Zo\odot\mathbf{X}\right)\mathbf{U}_{N}) + \mathbf{v}, 
\end{equation}
using \eqref{eqn:2D_CS_measure} and \eqref{eqn:msked_beam_R}. This measurement model is equivalent to a linear measurement model of the $N^2\times 1$ sparse vector $\mathbf{x}=\mathrm{vec}\left(\mathbf{X}\right)$. To represent the linear model explicitly, we replace the sampling operator $\mathcal{P}_{\Omega}(\cdot)$ with an $M\times N^2$ subsampling matrix $\mathbf{S}$ which is equal to $1$ at locations in $\{(m,Nc[m]+r[m])\}_{m=0}^{M-1}$ and zero at the remaining locations. Thus, $\mathbf{S}$ consists of exactly $|\Omega|=M$ ones representing the circular shift indices within $\Omega$. We use $\mathbf{z}_{\mathrm{o}}=\mathrm{vec}\left(\Zo\right)$ to denote the $N^2\times 1$ vector form of the spectral mask $\Zo$. The CS matrix that compresses $\mathbf{x}$ to $\mathbf{y}$ can be determined from \eqref{eqn:partial2DDFT} as 
\begin{equation}\label{eqn:cs_mat}
    \mathbf{A}_\mathrm{o} = \mathbf{S}\left(\mathbf{U}_N\otimes\mathbf{U}_N\right)\mathrm{Diag}(\mathbf{z}_\mathrm{o}).
\end{equation}
The linear measurement model in \eqref{eqn:2D_CS_measure} can be rewritten as 
\begin{equation}\label{eqn:1D_CS_measure}
    \mathbf{y} = \mathbf{A}_\mathrm{o}\mathbf{x} + \mathbf{v}.
\end{equation}
Note that the set of the circular shifts $\Omega = \{(r[m],c[m])\}_{m=0}^{M-1}$ determines the subsampling matrix $\mathbf{S}$ and therefore the CS matrix $\mathbf{A}_\mathrm{o}$.
\par The performance of greedy CS algorithms for sparse recovery usually depends on the mutual coherence of the CS matrix \cite{ben2010coherence}. The mutual coherence is simply the maximum of the normalized inner product between the columns of the CS matrix. It was shown in \cite{ben2010coherence} that minimizing the mutual coherence of the CS matrix results in a tight upper bound for the MSE in the sparse estimate. To this end, we optimize the set of the circulant shifts $\Omega$ to minimize the mutual coherence of the CS matrix in our problem.

\par The CS matrix in our in-sector channel estimation problem has a special structure due to our comb-like construction for the sectors. We define  $\mathcal{L}_\mathrm{o}=\{qN+p:(p,q)\in\Ao\}$ as the set of the 1D indices where the comb-like spectral mask $\Zo$ is known to be non-zero. Specifically, for the sector of interest $\Ao$, the vectorized spectral mask $\mathbf{z}_\mathrm{o}$ is non-zero at the locations in $\mathcal{L}_\mathrm{o}$ and zero at the other locations $[N^2]\backslash\mathcal{L}_\mathrm{o}$. Therefore, we can equivalently write the measurements \eqref{eqn:1D_CS_measure} as 
\begin{equation}\label{eqn:1D_effective_measure}
    \mathbf{y} = \mathbf{A}_{\mathcal{L}_\mathrm{o}}\mathbf{x}_{\mathcal{L}_\mathrm{o}} + \mathbf{v}.
\end{equation}
In \eqref{eqn:1D_effective_measure}, $\mathbf{A}_{\mathcal{L}_\mathrm{o}}$ is the $M\times \rho_{\mathrm{e}}\rho_{\mathrm{a}}$ submatrix of $\mathbf{A}_\mathrm{o}$ obtained by retaining the columns with indices in $\mathcal{L}_\mathrm{o}$, and $\mathbf{x}_{\mathcal{L}_\mathrm{o}}$ is a subvector of $\mathbf{x}$ with the indices in $\mathcal{L}_\mathrm{o}$. So, the in-sector CS problem is to estimate this $\rho_{\mathrm{e}}\rho_{\mathrm{a}} = N^2/S$ dimension sparse vector $\mathbf{x}_{\mathcal{L}_\mathrm{o}}$ from the $M$ measurements. Also, sparse recovery in our in-sector CS problem $\Ao$ depends on the coherence of $\mathbf{A}_{\mathcal{L}_\mathrm{o}}$ rather than $\mathbf{A}_\mathrm{o}$.
\par Now, we use the notion of the point spread function (PSF) in partial 2D-DFT CS problem \cite{lustig2007sparse} to aid our construction of the sampling set $\Omega$. We define the $N\times N$ binary matrix $\mathbf{N}_{\Omega}$ that is $1$ at the indices $\Omega=\{(r[m],c[m])\}_{m=0}^{M-1}$ and is zero at other indices. Therefore, $\mathbf{N}_{\Omega}$ has exactly $M$ entries with the value of $1$. The corresponding PSF is \cite{lustig2007sparse,myers2019falp}
\begin{equation}\label{eqn:psf}
    \mathbf{\mathrm{PSF}}_\mathrm{o} = \frac{N}{M}\mathbf{U}_N^{*}\mathbf{N}_{\Omega}\mathbf{U}_N^{*}.
\end{equation}
In a standard partial 2D-DFT CS problem, the coherence of the CS matrix is simply the maximum sidelobe level of the PSF. In our in-sector CS problem, however, only the sidelobes within the sector of interest matter. With our comb-like sectors constructed according to \eqref{eqn:comb_sec}, it can be shown that the coherence of the effective CS matrix $\mathbf{A}_{\mathcal{L}_\mathrm{o}}$ is
\begin{equation}\label{eqn:eff_coherence_2}
    \mu_\mathrm{o} = \underset{\{(i, j)\in\mathcal{T}\}}{\max}\ \ \lvert\mathbf{\mathrm{PSF}}_\mathrm{o}(i,j)\rvert,
\end{equation}
where $\mathcal{T}=\{(i,j):~i=nN_{\mathrm{e}},~j=mN_{\mathrm{a}},~(n,m)\in[\rho_{\mathrm{e}}]\times[\rho_\mathrm{a}]\backslash (0,0)\}$ for any comb-like sector of interest $\Ao$ of the form in \eqref{eqn:comb_sec}. The maximum sidelobe level of the PSF at the locations in $\mathcal{T}$ determines the mutual coherence of the CS matrix in our in-sector CS problem.

\par Now, we formulate the coherence minimization problem as
\begin{equation}\label{eqn:opt}
\mathcal{P}:
  \begin{cases}
    \begin{aligned}
        &\underset{\mathbf{N}_{\Omega}\in\{0,1\}^{N\times N}}{{\min}}
        && \underset{(i,j)\in\mathcal{T}}{{\max}} \ \ \ |\mathrm{PSF}_\mathrm{o}(i,j)|\\ 
        &\text{{s.t.}} && \sum_{(i,j)\in[N]\times[N]}N_{\Omega}(i,j)=M.
    \end{aligned}
  \end{cases}
\end{equation}
The problem $\mathcal{P}$ is non-convex and hard to solve. In Lemma \ref{lemma1}, however, we present an optimal solution $\mathbf{N}_{\Omega}^{\mathrm{Nyq}}$ that achieves $\mu_\mathrm{o}=0$ when $M=\rho_{\mathrm{e}}\rho_{\mathrm{a}}$. This case corresponds to the Nyquist sampling criterion where the number of measurements $M$ is equal to the sector dimension $N^2/S$. For the sub-Nyquist regime where $M<\rho_{\mathrm{e}}\rho_{\mathrm{a}}$, our randomized subsampling technique just selects $M$ elements at random from the Nyquist sampling set.
\begin{lemma}\label{lemma1}
For $M = \rho_{\mathrm{e}}\rho_{\mathrm{a}}$, the matrix $\mathbf{N}_{\Omega}^{\mathrm{Nyq}}$ corresponding to $\Omega = [\rho_{\mathrm{e}}]\times[\rho_{\mathrm{a}}]$, with $\rho_{\mathrm{e}}$ and $\rho_{\mathrm{a}}$ defined in \eqref{eqn:rho_e} and \eqref{eqn:rho_a}, is an optimal solution of $\mathcal{P}$ since it achieves $\mu_\mathrm{o}=0$ in \eqref{eqn:eff_coherence_2}.
\end{lemma}
\begin{proof}
See Section \ref{app:1a}. 
\end{proof}
\par We now discuss the Nyquist sampling criteria for the $N^2$ dimension channel estimation problem and the $N^2/S$ dimension in-sector channel estimation problem. In the former case, the Nyquist sampling criterion is $M=N^2$ and the optimal solution $\mathbf{N}_{\Omega}$ for $\mathcal{P}$ is an $N\times N$ matrix of all ones corresponding to $\Omega\! =\! [N]\times[N]$. The PSF corresponding to $\Omega = [N]\times[N]$ can be shown to be zero at all entries except $(0,0)$. For the in-sector CCS problem with our comb-like sectors, the ideal PSF need not be zero at all these locations to achieve $\mu_\mathrm{o}=0$. Instead, it only needs to be zero at the locations defined by $\mathcal{T}$ in \eqref{eqn:eff_coherence_2}. In Lemma \ref{lemma1}, we propose a construction for $\Omega$ that achieves $\mu_\mathrm{o}=0$ for $M=N^2/S$. This set simply comprises the 2D-coordinates of any $\rho_{\mathrm{e}} \times \rho_{\mathrm{a}}$ contiguous block on the $[N] \times [N]$ grid. 

\par We use $\mathbf{N}_{\Omega}^{\mathrm{Nyq}}$ to denote the optimal solution of $\mathcal{P}$ when $M=N^2/S$. In the sub-Nyquist regime, i.e., $M<\rho_{\mathrm{e}}\rho_{\mathrm{a}}$, we propose to select $M$ entries from the Nyquist set of circular shifts in Lemma \ref{lemma1}, i.e., $[\rho_{\mathrm{e}}]\times[\rho_{\mathrm{a}}]$, uniformly at random without replacement. We refer to this method as the proposed circular shifts (PCS) for $M<\rho_{\mathrm{e}}\rho_{\mathrm{a}}$, i.e., the subsampling set $\Omega$ in PCS is a subset of the optimal set in Lemma \ref{lemma1} for $M=\rho_{\mathrm{e}}\rho_{\mathrm{a}}$. Hence, PCS converges to the optimal solution proposed in Lemma \ref{lemma1} as $M \rightarrow \rho_{\mathrm{e}}\rho_{\mathrm{a}}$. We would like to mention that PCS is different from the fully random circular shifts (RCS) used in \cite{myers2019falp}. With RCS, the $M$ circular shifts are chosen at random from $[N]\times [N]$, while PCS selects the $M$ shifts at random from $[\rho_{\mathrm{e}}]\times[\rho_{\mathrm{a}}]$. 

\par To illustrate the effectiveness of our PCS in the sub-Nyquist regime $M < \rho_{\mathrm{e}}\rho_{\mathrm{a}}$ over RCS, we examine the mutual coherence of the CS matrices. In Fig. \ref{fig_6:PSFs}\textcolor{red}{(a)} and \ref{fig_6:PSFs}\textcolor{red}{(b)}, we compare the entry-wise magnitude of the PSF \eqref{eqn:psf} obtained from PCS and RCS in the Nyquist regime $M=N^2/S$. We observe that with PCS, the PSF is $0$ at the indices of interest, i.e., $\mathcal{T}$, while RCS which selects $M$ indices at random from $[N]\times[N]$ results in a higher coherence $\mu_\mathrm{o}$. Next, we plot the cumulative distribution function (CDF) of $\mu_\mathrm{o}$ for the two sampling schemes in Fig.~\ref{fig_6:PSFs}\textcolor{red}{(c)}. We notice from the CDF that PCS achieves a smaller coherence for the effective CS matrix $\mathbf{A}_{\mathcal{L}_\mathrm{o}}$ than RCS. Finally, we show an instance of PCS and RCS in Fig.~\ref{fig_7:sampling} to distinguish the difference between the two sampling schemes. 

\begin{figure}[t]
\centering
\subfloat[$|\mathrm{PSF}|$ for the proposed circulant shifts (PCS) with $\mu_{\mathrm{o}}= 0$ from \eqref{eqn:eff_coherence_2}.]{\includegraphics[trim=4.24cm 7cm 1.75cm 7cm, clip, scale=0.31]{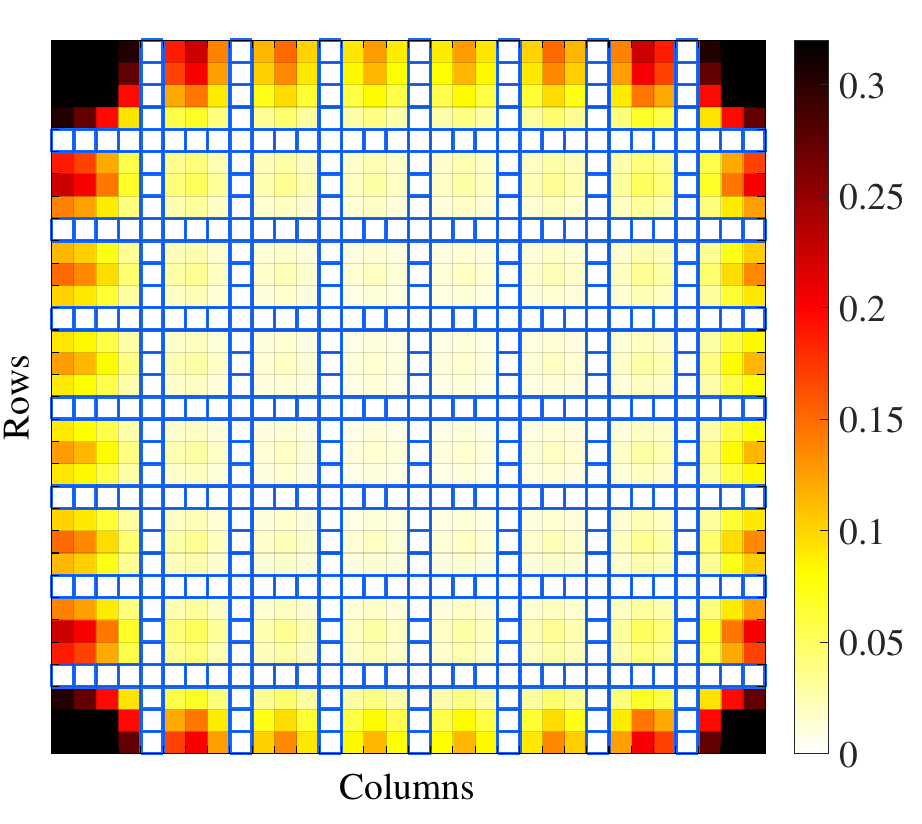}}
\hfil
\subfloat[$|\mathrm{PSF}|$ for the random circulant shifts (RCS) with $\mu_{\mathrm{o}}\approx 0.23$ from \eqref{eqn:eff_coherence_2}.]{\includegraphics[trim=4.24cm 7cm 1.75cm 7cm, clip, scale=0.31]{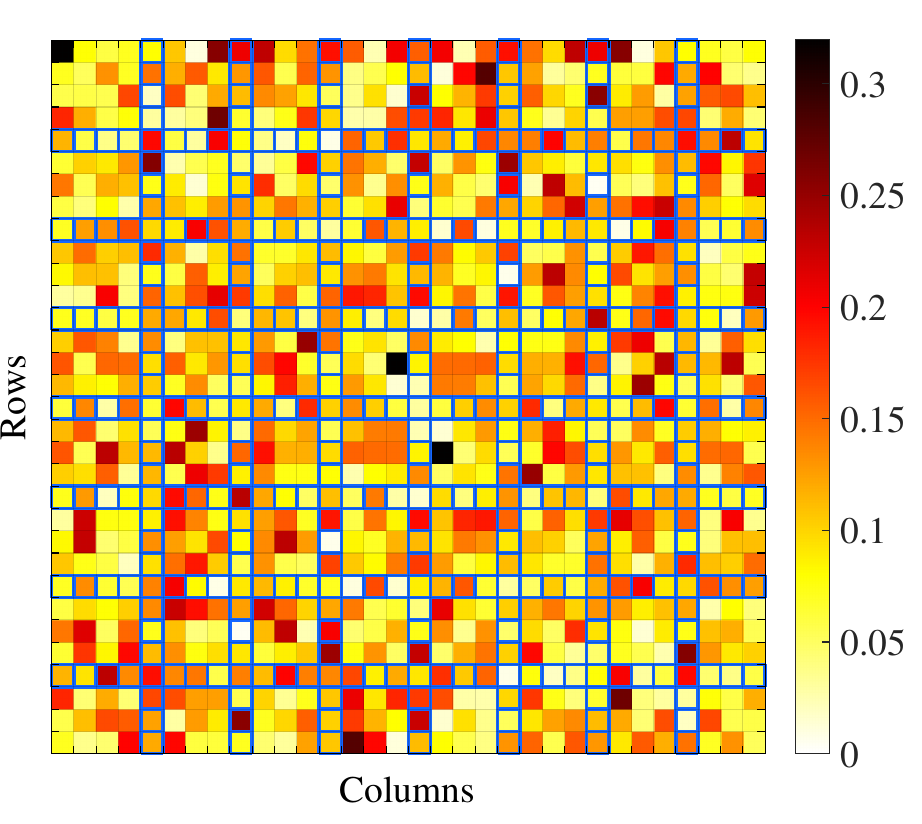}
}
\hfil
\subfloat[CDF of the coherence of the effective CS matrix $\mathbf{A}_{\mathcal{L}_\mathrm{o}}$ for PCS and RCS.]{\includegraphics[trim=1.75cm 6.6cm 2.25cm 7.45cm, clip, scale=0.31]{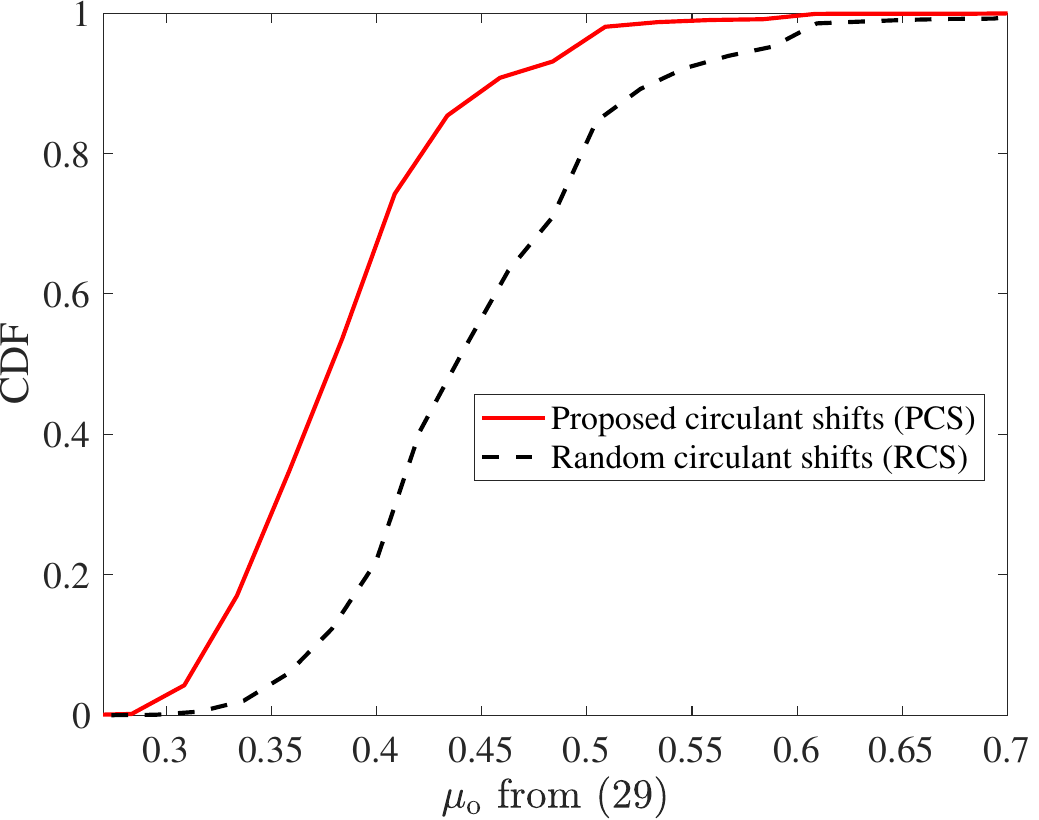}
}
\caption{\small Proposed subsampling and fully random subsampling for $N=32$, $N_{\mathrm{e}}=4$, $N_{\mathrm{a}}=4$, $S=N_{\mathrm{e}}N_{\mathrm{a}}$, $\rho_{\mathrm{e}}=N/N_{\mathrm{e}}$, $\rho_{\mathrm{a}}=N/N_{\mathrm{a}}$. In (a) and (b), we assume $M=\rho_{\mathrm{e}}\rho_{\mathrm{a}}$ and the entries indicated by \textcolor[rgb]{0.0824 0.4039 1}{\scriptsize{$\square$}} have indices in $\mathcal{T}$. The PSF for PCS achieves $\mu_{\mathrm{o}}=0$, which corresponds to zero aliasing artifacts under Nyquist sampling. For the same number of measurements, however, the PSF with RCS results in a positive $\mu_{\mathrm{o}}$. We observe from (c) that the coherence with PCS is lower than that with RCS in the sub-Nyquist regime for $M=20$. \normalsize}\label{fig_6:PSFs} 
\end{figure}
\begin{figure}[t]
\centering
{\includegraphics[trim=2.4cm 5.6cm 3.9cm 0.7cm, clip, scale=0.48]{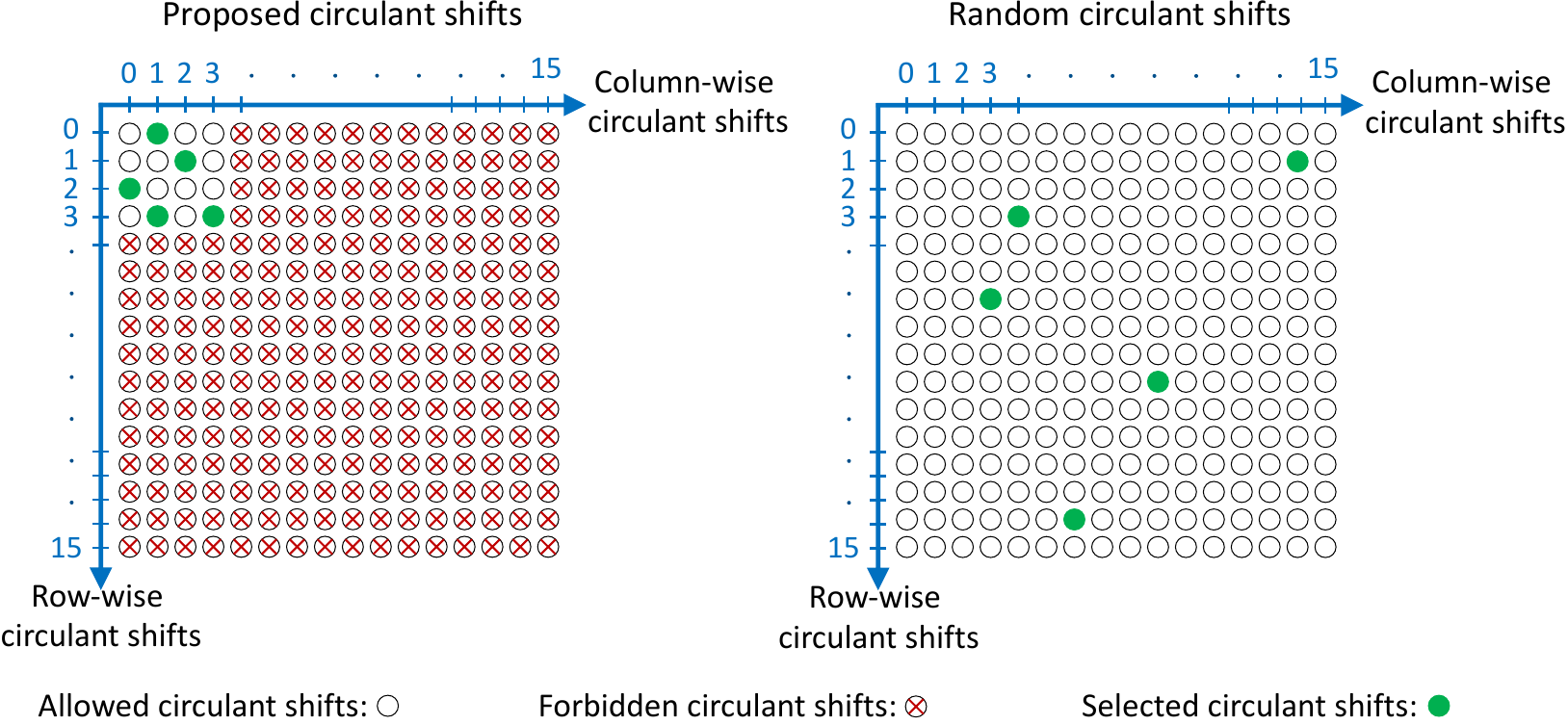}}
\caption{\small Our PCS scheme for $N=16$, $N_{\mathrm{e}}=4$, $N_{\mathrm{a}}=4$, $S=N_{\mathrm{e}}N_{\mathrm{a}}$, $\rho_{\mathrm{e}}=N/N_{\mathrm{e}}$, $\rho_{\mathrm{a}}=N/N_{\mathrm{a}}$, and $M=5$. Here, the circular shifts in PCS and RCS are chosen at random from $\{0,1,2,3\}\times\{0,1,2,3\}$ and $[16]\times[16]$ respectively.\normalsize}\label{fig_7:sampling} 
\end{figure}
\subsection{Guarantees on in-sector CS with the OMP}\label{sec6_mse}
In this section, we discuss our guarantees for OMP-based sparse channel estimation within the sector of interest $\Ao$. Our guarantees are an extension of the coherence-based guarantees in \cite{ben2010coherence} and they help us study the impact of variations in the illumination pattern within the sector. Here, we only provide the main results and we refer the interested reader to our work in \cite{masoumi2023analysis} for more details. 
 
\par Prior work on coherence-based guarantees for the OMP makes a strong assumption that the CS matrix has equal column norms. In our problem, the norms of the columns within the CS matrix are controlled by $\mathbf{z}_{\mathrm{o}}$ as observed from \eqref{eqn:cs_mat}. To express our guarantees as a function of the variations in  $\mathbf{z}_{\mathrm{o}}$, we define $\mathbf{z}_{\mathcal{L}_\mathrm{o}}$ as a $\rho_\mathrm{e}\rho_{\mathrm{a}}\times 1$ subvector of the vectorized spectral mask $\mathbf{z}_\mathrm{o}$ obtained by retaining indices in $\mathcal{L}_\mathrm{o}$. We use $d_i$ to denote the $\ell_2$-norm of the $i^{\mathrm{th}}$ column of $\mathbf{A}_{\mathcal{L}_\mathrm{o}}$. Using the definition of $\mathbf{A}_\mathrm{o}$ in \eqref{eqn:cs_mat}, we can write that 
\begin{equation}\label{eqn:l2_norm1}
    d_i = \frac{\sqrt{M}}{{N}}|z_{\mathcal{L}_\mathrm{o}}(i)|,~~~~\forall i\in[\rho_{\mathrm{e}}\rho_{\mathrm{a}}],
\end{equation}
where $z_{\mathcal{L}_\mathrm{o}}(i)$ denotes the $i^{\mathrm{th}}$ entry of $\mathbf{z}_{\mathcal{L}_\mathrm{o}}$. The maximum and the minimum of the column norms in $\mathbf{A}_{\mathcal{L}_\mathrm{o}}$ are defined as $d_{\mathrm{max}} = \underset{j\in[\rho_{\mathrm{e}}\rho_{\mathrm{a}}]}{\max}~d_{j}$ and $d_{\mathrm{min}} = \underset{j\in[\rho_{\mathrm{e}}\rho_{\mathrm{a}}]}{\min}~d_{j}$.
As the $\ell_2$-norm of $\mathbf{z}_{\mathcal{L}_\mathrm{o}}$ is equal to the Frobenius norm of the AWM $\Po$, i.e., $1$, we have $\sum_{i=1}^{\rho_{\mathrm{e}}\rho_{\mathrm{a}}}d_i^{2}=M/N^2$. Due to this energy constraint, 
\begin{equation}\label{eqn:dmin_max}
    0<d_{\mathrm{min}} \leq d_{\mathrm{max}}<\sqrt{M}/N.
\end{equation}
The entries of the spectral mask in $\mathbf{z}_{\mathcal{L}_\mathrm{o}}$ are proportional to the norm of the columns of the CS matrix.
\par Now, we provide performance guarantees to recover the sparse in-sector beamspace $\mathbf{x}_{\mathcal{L}_\mathrm{o}}$ from the noisy CS measurements in \eqref{eqn:1D_CS_measure} using the OMP \cite{gharavi1998fast}. The guarantees provided in Theorem \ref{theorem1} are for successful support recovery and for the MSE in the reconstruction.
\begin{theorem}\label{theorem1}
Let $x_{\mathcal{L}_\mathrm{o},\mathrm{min}} = \underset{j\in\Pi}{\min}~|x_{\mathcal{L}_\mathrm{o},j}|$, where $\Pi$ denotes the support of $\mathbf{x}_{\mathcal{L}_\mathrm{o}}$ with cardinality $|\Pi|=k$. If
\begin{equation}\label{eqn:cond_min}
    d_{\mathrm{min}}x_{\mathcal{L}_\mathrm{o},\mathrm{min}} - (2k-1)\mu d_{\mathrm{max}}x_{\mathcal{L}_\mathrm{o},\mathrm{min}} \geq 2\gamma,
\end{equation}
for $\gamma>0$, with probability exceeding
\begin{equation}\label{eqn:lemma1}
        \mathrm{Pr}\!\left\lbrace E \right\rbrace \geq \left(1-\sqrt{\frac{2}{\pi}}.\sqrt{\frac{\sigma}{\gamma}}\exp{\!\left(\!-\frac{\gamma^2}{2\sigma^2}\right)}\right)^{2N^2/S}.
    \end{equation}
the OMP algorithm successfully recovers the support of $\mathbf{x}_{\mathcal{L}_\mathrm{o}}$ and the MSE of the estimate $\hat{\mathbf{x}}_{\mathcal{L}_\mathrm{o}}$ is upper bounded as
\begin{equation}\label{eqn:mse_est}
        \|\hat{\mathbf{x}}_{\mathcal{L}_\mathrm{o}}-\mathbf{x}_{\mathcal{L}_\mathrm{o}}\|_2^2 \leq \left(\frac{d_{\mathrm{max}}}{d_{\mathrm{min}}}\right)^{2}\frac{k\gamma^2}{\left(d_{\mathrm{min}} - (k-1)\mu d_{\mathrm{max}}\right)^2}.
    \end{equation}
\end{theorem}
\begin{proof}
The proof is discussed in \cite{masoumi2023analysis}.
\end{proof}
\par An important insight from Theorem \ref{theorem1} is that the OMP can identify weak coefficients in the beamspace when $d_{\mathrm{max}}/d_{\mathrm{min}}$ of the CS matrix is small. This observation follows by rewriting \eqref{eqn:cond_min} as 
\begin{align}
    x_{\mathcal{L}_\mathrm{o},\mathrm{min}} &{\geq} \frac{2\gamma}{d_{\mathrm{min}} - (2k-1)\mu d_{\mathrm{max}}}\\
    &= \frac{2\gamma/d_{\mathrm{min}}}{1 - \mu(2k-1) (d_{\mathrm{max}}/d_{\mathrm{min}})}.
\end{align}
The upper bound on the MSE in \eqref{eqn:mse_est} monotonically increases with $d_{\mathrm{max}}/d_{\mathrm{min}}$. In summary, our guarantees on support recovery and the MSE become tight for a small  $d_{\mathrm{max}}/d_{\mathrm{min}}$. We notice from \eqref{eqn:dmin_max} that the smallest possible $d_{\mathrm{max}}/d_{\mathrm{min}}$ is $1$. This is achieved when $\underset{i\in[\rho_{\mathrm{e}}\rho_{\mathrm{a}}]}{\min}~ |z_{\mathcal{L}_\mathrm{o}}(i)|=\sqrt{S}/N$, equivalently when $\Po$ uniformly illuminates the sector of interest. These guarantees motivate our optimization objective in \eqref{eqn:opt_weight} which aims to generate AWMs that result in an almost uniform illumination within the sector. Further, it can be concluded from our guarantees that the illumination pattern in Fig.~\ref{fig_6:SpectralMaskWeights}\textcolor{red}{(a)} is expected to achieve better in-sector channel reconstruction than the pattern in Fig.~\ref{fig_6:SpectralMaskWeights}\textcolor{red}{(b)}.  
\section{Numerical Results}\label{sec5sims}
In this section, we first describe how the proposed beamforming method is extended to a wideband system using the IEEE 802.11ad frame structure. Then, we discuss the performance with our proposed method in terms of the normalized MSE in the in-sector channel estimate and the achievable rate.
\begin{figure}[t]
\centering
\includegraphics[trim=0.2cm 5.9cm 0.4cm 5.92cm, clip, scale=0.49]{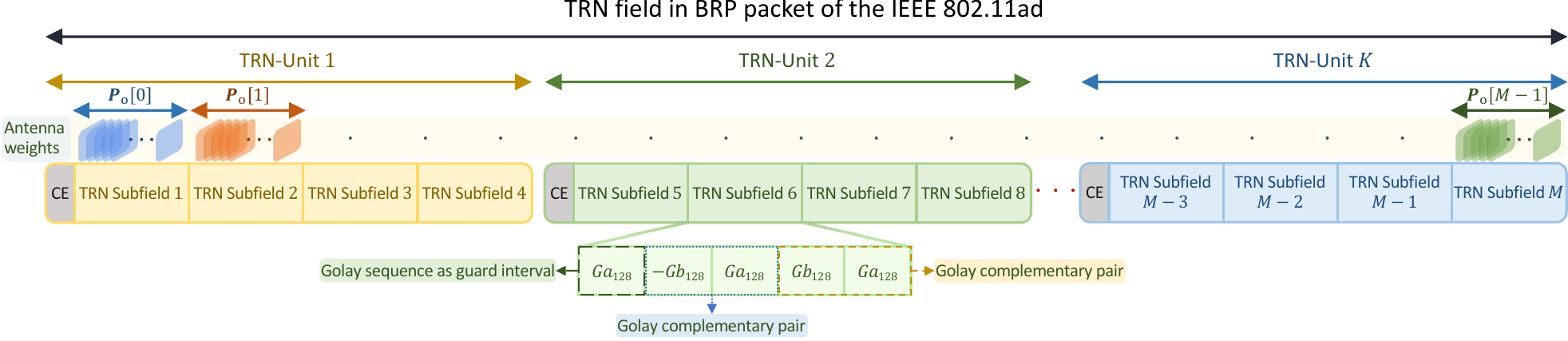}
\caption{\small The TRN field in a BRP packet of the IEEE 802.11ad standard consists of $K$ TRN-Units, where $K$ can be as large as $32$ \cite[Table 20-13]{standard}. Each TRN-Unit contains $4$ TRN subfields comprising Golay complementary pairs \cite[Sec. 20.9.2.2.6]{standard}, \cite{golay1961complementary}. The in-sector CS measurements can be acquired within the sector of interest $\Ao$, when the TX sequentially applies $\{\Po[m]\}_{m=0}^{M-1}$ at the beginning of each TRN subfield. \normalsize}\label{fig:brp} 
\end{figure}
\subsection{Extension to a wideband system}
\par We denote the $L-$tap wideband channel as $\{\mathbf{H}[\ell]\}_{\ell=0}^{L-1}$ where $\mathbf{H}[\ell]\in\mathbb{C}^{N\times N}$. The TX first identifies the best sector by applying $S$ different AWMs $\{\mathbf{P}_s\}_{s=0}^{S-1}$ in SLS. The best sector is the one for which the received power, i.e., $\sum_{\ell=0}^{L-1}|\left\langle \mathbf{H}[\ell], \mathbf{P}_s \right\rangle + v_{\ell}|^2$ is the highest and the corresponding base AWM is $\Po$. Then, the TX acquires CS measurements within this sector by applying different AWMs $\{\Po[m]\}_{m=0}^{M-1}$ at the beginning of each TRN subfield in a BRP packet. As shown in Fig.~\ref{fig:brp}, to obtain the $m^{\mathrm{th}}$ measurement, at the beginning of the $m^{\mathrm{th}}$ TRN subfield, TX sets the antenna weights at the UPA to $\Po[m]$. With our design, $\Po[m]$ is a 2D-circular shift of $\Po$. 
\par In IEEE 802.11ad, a Golay sequence \cite{golay1961complementary} $Ga_{128}$ of length $128$ is used as a guard interval at the beginning of the TRN subfield. This interval provides enough time to change the antenna weights at the TX and also avoids interference across successive subfields. Then, the TX transmits two Golay complementary pairs $\{-Gb_{128},Ga_{128}\}$ and $\{Gb_{128},Ga_{128}\}$ for each AWM $\{\Po[m]\}_{m=0}^{M-1}$ applied at the UPA. The complementary property of Golay sequences can be used to obtain measurements of the equivalent wideband channel seen through the beamformer $\Po[m]$ \cite{mishra2017sub}. This channel is $\{\langle\mathbf{H}[\ell],\Po[m]\rangle\}_{\ell=0}^{L-1}$, a vector of length $L$ for each $m\in[M]$. We define $\mathbf{Y}_{\mathrm{blk}}\in\mathbb{C}^{M\times L}$ to denote the set of these $ML$ spatial channel projections within the sector of interest and the measurement noise as $\mathbf{V}_{\mathrm{blk}}\in\mathbb{C}^{M\times L}$. Hence, the $(m, \ell)^{\mathrm{th}}$ entry of $\mathbf{Y}_{\mathrm{blk}}$ is
\begin{equation}\label{eqn:wide}
   Y_{\mathrm{blk}}(m,\ell) = \langle\mathbf{H}[\ell],\Po[m]\rangle + V_{\mathrm{blk}}(m, \ell).
\end{equation}
Due to the spreading gain of $N_\mathrm{seq}=256$ with the Golay sequence correlators, the entries of $\mathbf{V}_{\mathrm{blk}}$ are independently and independently distributed as $\mathcal{CN}(0,\sigma^2/N_\mathrm{seq})$.
\par In this paper, only the DC subcarrier, i.e., the sum of the $L$ taps of the wideband channel $\mathbf{H}_{\mathrm{sum}}=\sum_{\ell=0}^{L-1}\mathbf{H}[\ell]$ is estimated within the sector of interest. Although our method can be applied to estimate all the taps, we only estimate the sum of the channel taps as it requires less computational complexity than wideband channel estimation. This approach is reasonable because the phased array can only apply a frequency-flat beamformer, which may be designed based on a single subcarrier. The in-sector CS measurements are thus $y[m]=\sum_{\ell=0}^{L-1}Y_{\mathrm{blk}}(m,\ell),~\forall m\in[M]$. Now, similar to the narrowband case, in-sector CS can be performed using $\{y[m]\}_{m=0}^{M-1}$ to obtain an estimate of the DC subcarrier $\hat{\mathbf{H}}_{\mathrm{sum}}^{\mathrm{o}}$ within the sector of interest. This estimate can then be used to design the transmit beamformer. 
\subsection{System parameters}
We consider a $32\times 32$ phased array in Fig.~\ref{fig_1:sysmdl}\textcolor{red}{(a)} at the TX and we assume $1$-bit phase shifters unless otherwise stated. The TX-RX distance is 60 $\mathrm{m}$. The heights of the TX and the RX are $3~\mathrm{m}$ and $1.5~\mathrm{m}$. The system operates at a carrier frequency of $60~\mathrm{GHz}$ with a bandwidth of $100~\mathrm{MHz}$ that corresponds to a symbol duration of $10~\mathrm{ns}$. We use $\mathrm{SNR}_{\mathrm{omni}}$ to denote the SNR without spreading gain and without any beamforming gain. This is obtained when the TX applies a quasi-omnidirectional beamformer generated by a perfect binary array \cite{myers2019falp}. We use mmWave channels obtained from the NYU channel simulator \cite{sun2017novel} for a line-of-sight scenario in an urban micro environment. For $100$ independent channel realizations, the omnidirectional RMS delay spread was less than $25~\mathrm{ns}$ in more than $90\%$ of the channel realizations. We model the wideband channel using $L=10$ taps corresponding to a duration of $100~\mathrm{ns}$ which is much larger than $25~\mathrm{ns}$. The simulation results in this section are averaged over the $100$ channel realizations. 
\subsection{Performance of the proposed in-sector CS-based beamforming}
\par We discuss metrics used to evaluate the performance of the proposed method against two other benchmarks. We also assess the performance gain arising from each of our three contributions: (i) Construction of non-overlapping comb-like sectors in SLS, (ii) Optimization of the weight matrix $\mathbf{W}^s$ to achieve an almost uniform illumination pattern, and (iii) Optimization of the circular shifts in PCS for in-sector CS.

\par To construct the base AWMs in our approach for SLS, we initialize the weight matrix within the iterative optimization algorithm in \eqref{eqn:opt_weight} to different matrices that exhibit good autocorrelation properties. This is done because the algorithm can converge to different local minima \cite{zhao2016unified} depending on the initialization. To this end, we use perfect binary arrays \cite{jedwab1994perfect}, the quantized DFT matrix, the outer product of two Frank sequences \cite{frank1962phase}, the outer product of two Zadoff-Chu (ZC) sequences \cite{chu1972polyphase}, and the outer product of two Golomb sequences \cite{zhang1993polyphase} for the initialization. Then, the initialization that results in the most uniform magnitude profile within our comb-like construction is used to construct a base AWM for the sector. The sectors generated with our approach have a comb-like structure, unlike the contiguous sectors discussed in prior work.
\par For CS-based channel estimation within the sector of interest, we used the OMP algorithm \cite{gharavi1998fast} to obtain an estimate of $\mathbf{H}_{\mathrm{sum}}$. We study the performance of 2D-CCS-based with both RCS and the optimized PCS. We note that in RCS, it is possible to have the same AWM $\Po[m]$ used to obtain measurements. We remove such measurements to avoid rank deficiency in OMP sparse recovery. Let $\mathbf{H}_{\mathrm{sum}}^{\mathrm{o}}$ denote the spatially filtered version of $\mathbf{H}_{\mathrm{sum}}$ within the sector of interest. The matrix $\mathbf{H}_{\mathrm{sum}}^{\mathrm{o}}$ is obtained by applying a binary mask, that is one at the indices in $\Ao$, over the beamspace representation of $\mathbf{H}_{\mathrm{sum}}$. We then define the normalized in-sector MSE $\mathbb{E}[\|\mathbf{H}_{\mathrm{sum}}^{\mathrm{o}}-\hat{\mathbf{H}}_{\mathrm{sum}}^{\mathrm{o}}\|^2_{\mathrm{F}}]/\mathbb{E}[\|\mathbf{H}_{\mathrm{sum}}^{\mathrm{o}}\|^2_{\mathrm{F}}]$
as a metric for the reconstruction error in $\hat{\mathbf{H}}_{\mathrm{sum}}^{\mathrm{o}}$. 
\par After the in-sector channel is estimated, the TX applies a beamformer $\mathbf{F}\in\mathbb{Q}_{q}^{N\times N}$ such that $\lvert\langle\mathbf{F},\hat{\mathbf{H}}_{\mathrm{sum}}^{\mathrm{o}}\rangle\rvert$ is maximized. Let $\mathrm{phase(\kappa)}$ denote the phase of the complex scalar $\kappa$. When there is no constraint on the resolution of the phase shifters, i.e., $q=\infty$, from the dual norm inequality \cite{boyd2004convex}, $F^{\mathrm{opt}}_{ij}=\exp\left(\mathsf{j}{\mathrm{phase}([\hat{H}_{\mathrm{sum}}^{\mathrm{o}}]_{ij})}\right)/N$ is the maximizer of $\lvert\langle\mathbf{F}^{\mathrm{opt}},\hat{\mathbf{H}}_{\mathrm{sum}}^{\mathrm{o}}\rangle\rvert$. Under $q-$bit phase quantization, the TX applies $\mathbf{F}=\mathcal{Q}_q(\mathbf{F}^{\mathrm{opt}})$ as the transmit beamformer. The effective wideband single-input single-output channel is then $\{\langle\mathbf{H}[\ell],\mathbf{F}\rangle\}_{\ell=0}^{L-1}$. Finally, we use the \textit{waterfilling} algorithm to obtain the achievable rate associated with the equivalent channel\cite{tse2005fundamentals}.
\par For the benchmarks, we use the methods from \cite{tsai2018structured,ali2017millimeter} that use contiguous sectors similar to Fig.~\ref{fig_2:sectorTypes}\textcolor{red}{(a)}. These sectors are defined by $\mathcal{B}_s = \{(p,q): k_{\mathrm{e}}\rho_{\mathrm{e}}\leq p < (k_{\mathrm{e}}+1)\rho_{\mathrm{e}},~ k_{\mathrm{a}}\rho_{\mathrm{a}}\leq q < (k_{\mathrm{a}}+1)\rho_{\mathrm{a}}\}$, that each represent a $\rho_{\mathrm{e}} \times \rho_{\mathrm{a}}$ block in the beamspace. The AWMs from \cite{tsai2018structured} are based on ZC sequences. Specifically, the $M$ in-sector CS measurements in \cite{tsai2018structured} are obtained by modulating a base AWM with $M$ columns of the 2D-DFT dictionary, that are chosen at random. Implementing such an AWM requires $\mathrm{log}_2 N$ resolution phase shifters, which may not be available in practice. To this end, we quantize the phase of the codewords to $q-$bits. For the greedy method in \cite{ali2017millimeter}, we generate one million AWMs from $\mathbb{Q}_q^{N \times N}$ at random. Then, for each sector in $\{\mathcal{B}_s\}_{s=0}^{S-1}$, we select the top AWM that radiates the largest energy within the sector to be used during the SLS. After determining the sector of interest in SLS, the top $M$ AWMs that radiate the largest energy within the sector of interest are selected to obtain in-sector CS measurements with \cite{ali2017millimeter}.
\begin{figure}[t]
\centering
\subfloat[]{
\includegraphics[trim=1.5cm 6.5cm 2.5cm 7.45cm, clip, scale=0.42]{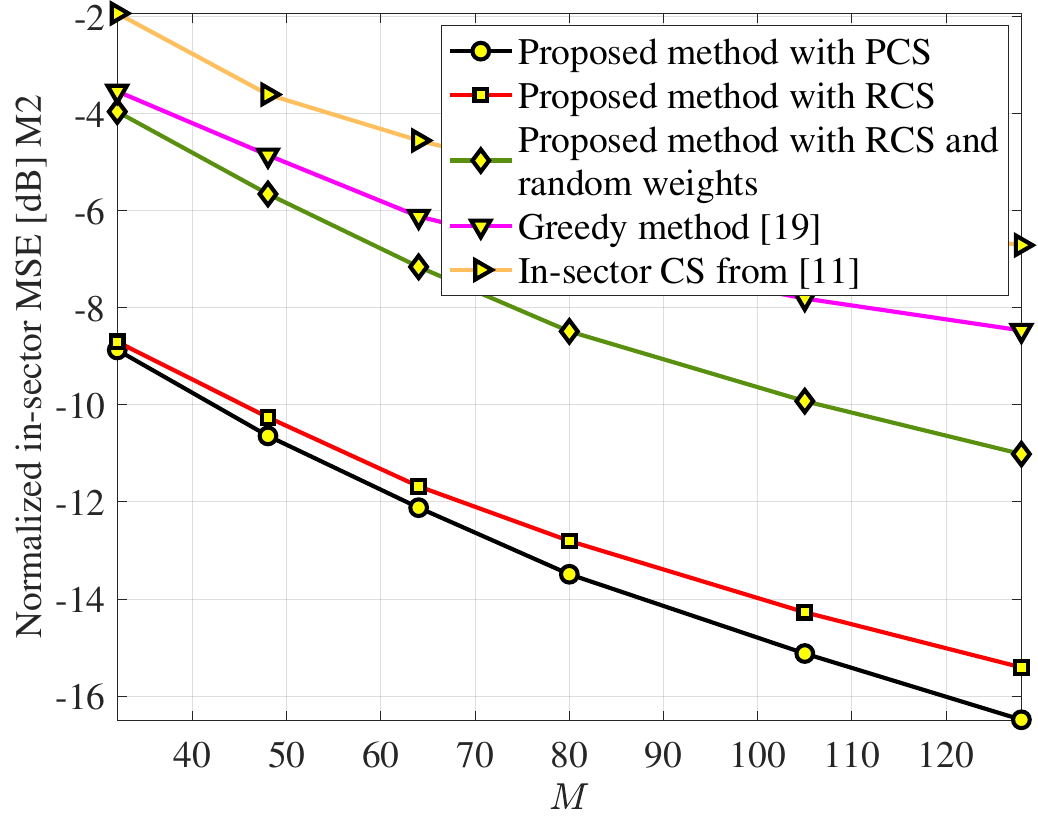}}
\hfil
\subfloat[]{
\includegraphics[trim=1.5cm 6.5cm 2.5cm 7.45cm, clip, scale=0.42]{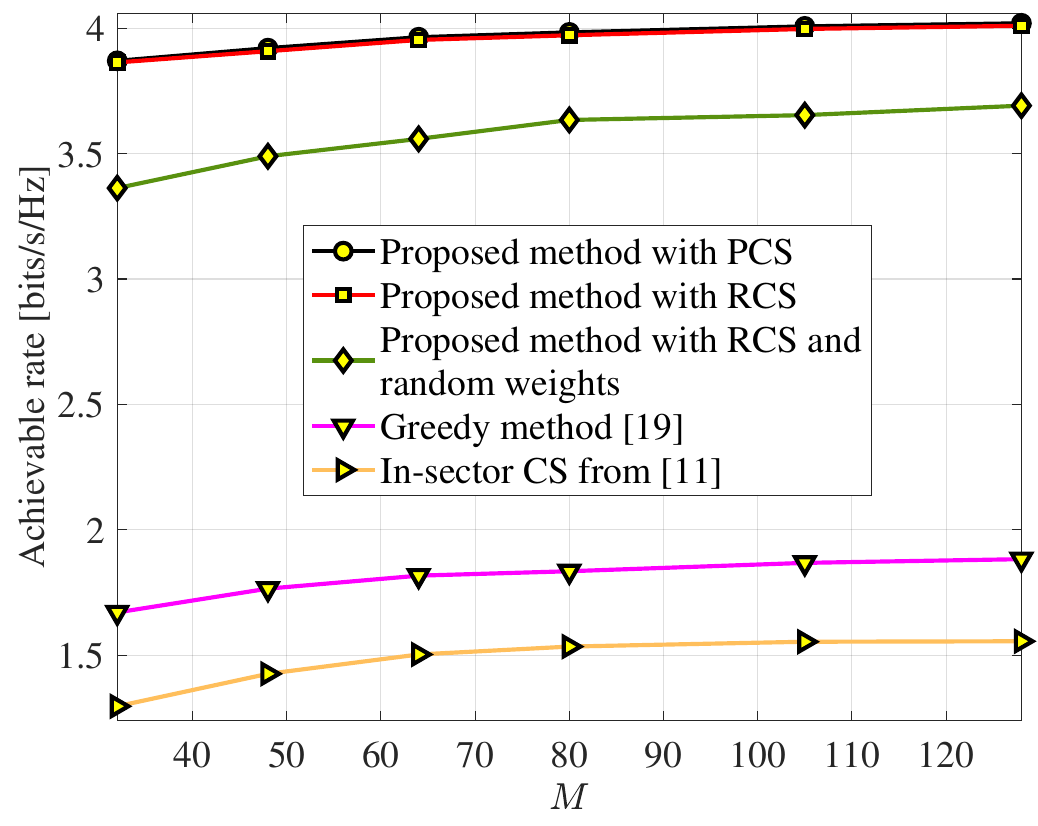}
}
\hfil
\subfloat[]{
\includegraphics[trim=1.5cm 6.5cm 2.5cm 7.45cm, clip, scale=0.42]{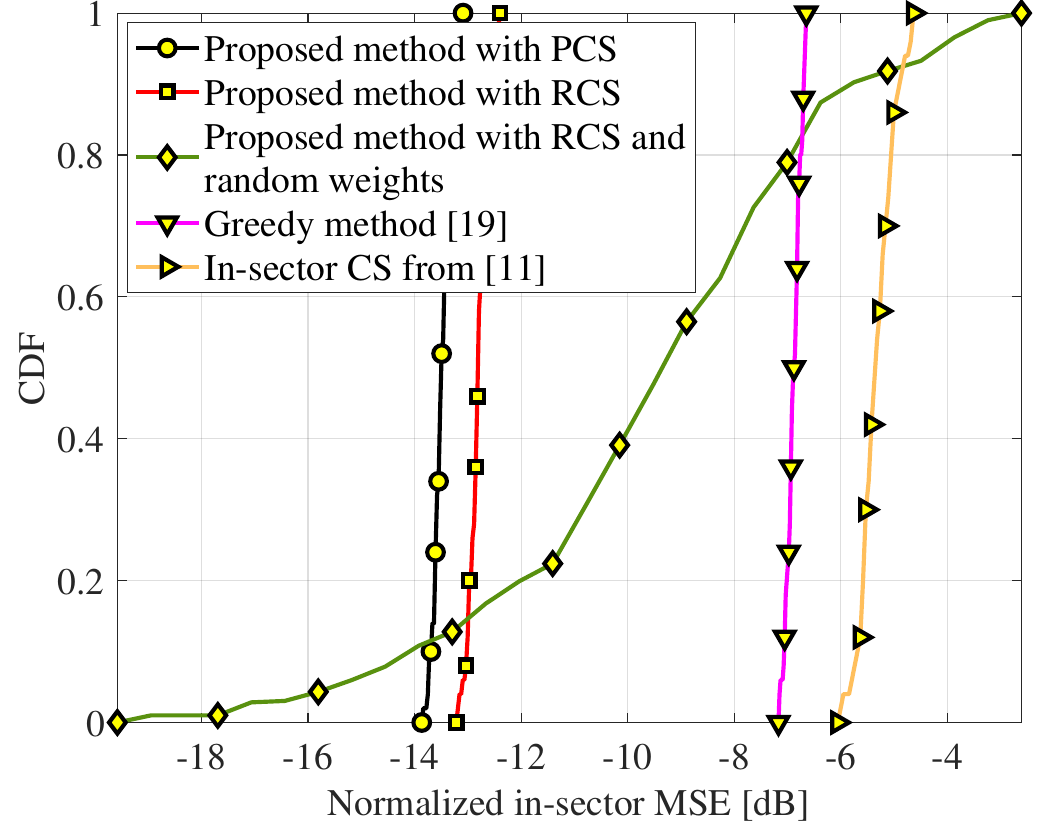}
}
\hfil
\subfloat[]{
\includegraphics[trim=1.5cm 6.5cm 2.5cm 7.45cm, clip, scale=0.42]{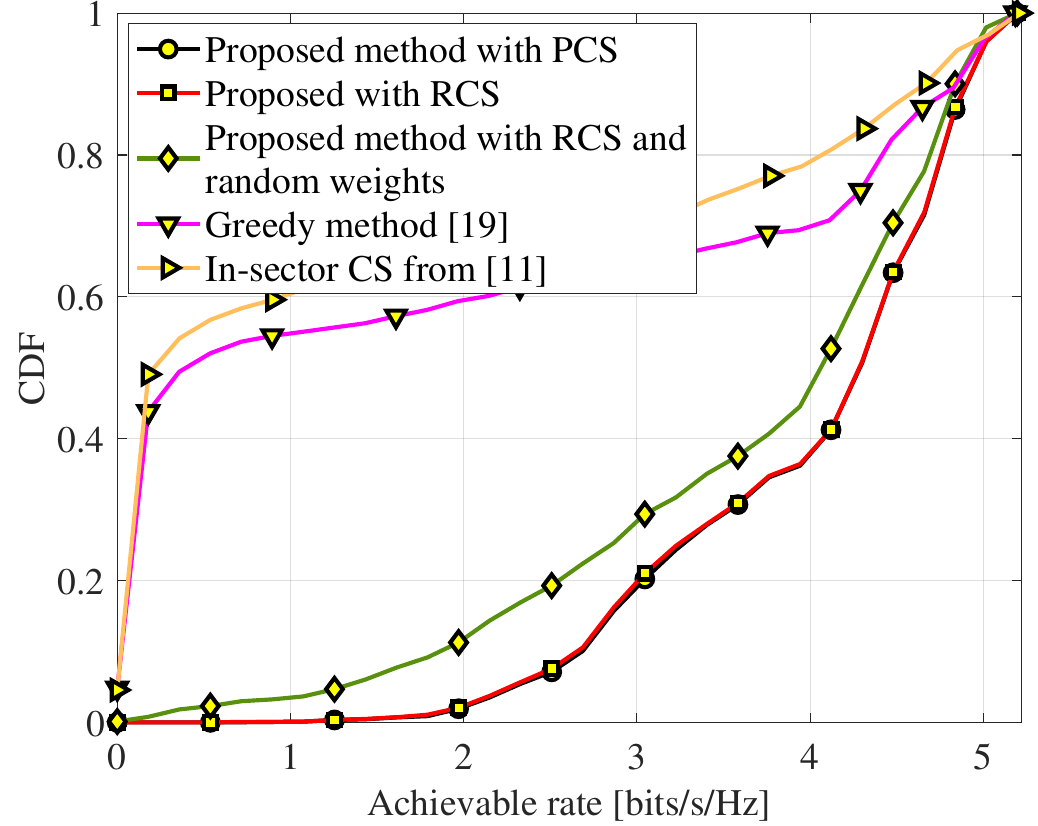}
}
\caption{\small  The plots show the normalized in-sector MSE,  achievable rate, and the CDF of the achievable rate using the proposed in-sector CCS-based method as well as the benchmarks. We assume $N_\mathrm{e}=2$, $N_\mathrm{e}=2$ which results $S=4$ sectors. We use $\mathrm{SNR}_{\mathrm{omni}}=-10~\mathrm{dB}$. In (c) and (d), we use $M=80$. It can be noticed that the proposed in-sector CCS-based method  outperforms the benchmarks when the weight matrix $\mathbf{W}^s$ in \eqref{eqn:AWM_from_upsampled} is optimized.\normalsize}\label{fig_9:vsM} 
\end{figure}
\par In Fig.~\ref{fig_9:vsM}\textcolor{red}{(a)}, we compare the normalized in-sector MSE of the estimated channel for the proposed in-sector CCS-based method and the benchmarks. As we observe, there is a large gap between the proposed method with PCS or RCS and the benchmarks \cite{tsai2018structured,ali2017millimeter} for a wide range of CS measurements. This is because our developed AWMs result in higher received power within the sector of interest, boosting the received SNR and therefore resulting in a low normalized in-sector MSE. The high received power within the sector of interest in our method is due to the concentration of the transmitter's energy only within the sector while the AWMs in \cite{tsai2018structured,ali2017millimeter} lead to a substantial amount of power leakage outside the sector of interest. Also, we observe that our method with PCS results in a lower normalized in-sector MSE than with RCS within the sector of interest. This is because PCS leads to lower aliasing artifacts within the sector of interest than the RCS as we observed in Fig.~\ref{fig_6:PSFs}. We can also observe a similar performance difference between our proposed method using PCS or RCS and the benchmarks in the achievable rate from Fig.~\ref{fig_9:vsM}\textcolor{red}{(b)} for different numbers of the measurement $M$. 
\par From the CDF plot in Fig.~\ref{fig_9:vsM}\textcolor{red}{(d)}, we observe that with a probability exceeding $0.98$, our proposed method with the optimized weight matrix $\mathbf{W}^{s}$ results in an achievable rate of at least about $2$ bits/s/Hz. The use of random weights for $\mathbf{W}^{s}$ results in a lower rate, which motivates the need to use almost uniform comb-like beams designed in this work. For the other in-sector CS benchmarks based on \cite{tsai2018structured,ali2017millimeter}, we observe that the achievable rate is lower than that with our approach. This is because these methods radiate a significant amount of the transmitter's energy outside their intended sectors. As a result, our solutions achieve lower in-sector NMSE and a higher rate than the benchmarks in \cite{tsai2018structured,ali2017millimeter}.
\begin{figure}[t]
\centering
\subfloat[]{
\includegraphics[trim=1.5cm 6.6cm 2.5cm 7.45cm, clip, scale=0.45]{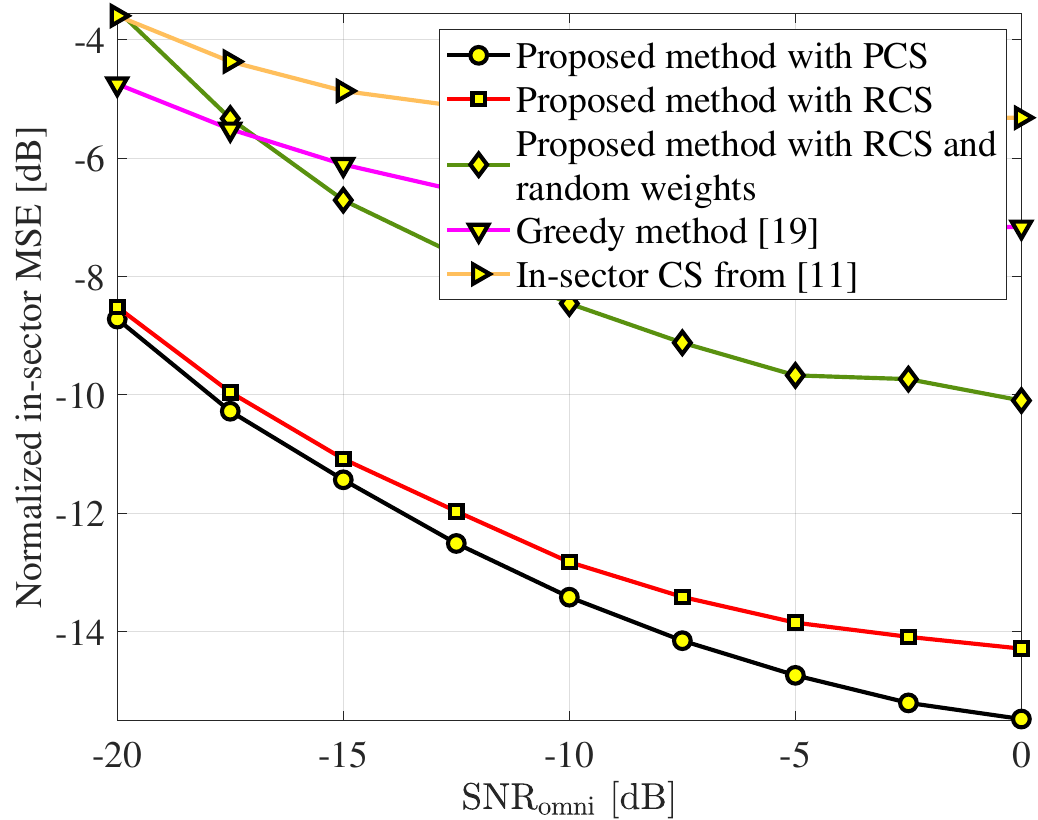}}
\hfil
\subfloat[]{
\includegraphics[trim=1.5cm 6.6cm 2.5cm 7.45cm, clip, scale=0.45]{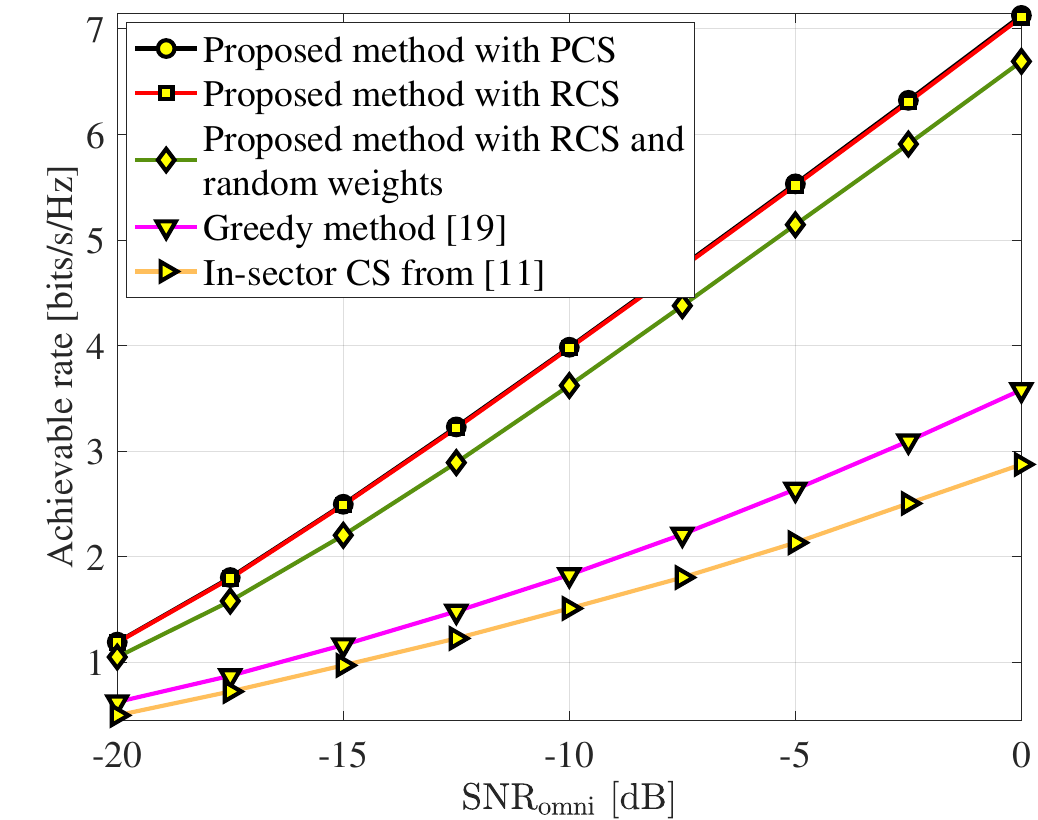}
}
\caption{\small  Here, we consider $N_\mathrm{e}=2$, $N_\mathrm{e}=2$ that correspond to $S=4$ sectors, and use $M=80$ in-sector CS measurements. For in-sector CS with our optimized comb-like sectors, we observe that PCS achieves a lower NMSE than RCS. This is because the designed set of circular shifts results in lower aliasing artifacts than random shifts. Furthermore, our method outperforms the benchmarks in \cite{tsai2018structured,ali2017millimeter} in terms of the NMSE and the rate.
\normalsize}\label{fig_9:vsSNR} 
\end{figure}
\par We observe from Fig.~\ref{fig_9:vsM}\textcolor{red}{(a)} that the proposed method with RCS and a random $\mathbf{W}^s$ results in a higher normalized in-sector MSE than the case with the optimized $\mathbf{W}^s$. This is because a random $\mathbf{W}^s$ is very likely to result in a  non-uniform illumination pattern within the comb-like structure as shown in Fig.~\ref{fig_6:SpectralMaskWeights}\textcolor{red}{(b)}. The column norms of such a CS matrix are substantially different with such an illumination, which is expected to result in poor reconstruction according to our analysis in Section \ref{sec6_mse}.  
\par In Fig.~\ref{fig_9:vsSNR}, we compare the proposed in-sector CS method against the benchmarks for different values of $\mathrm{SNR}_{\mathrm{omni}}$. We observe that our method outperforms the benchmarks over a wide range of $\mathrm{SNR}_{\mathrm{omni}}$. Furthermore, it is evident that at high $\mathrm{SNR}_\mathrm{omni}$, PCS outperforms RCS due to a significantly lower in-sector NMSE in the channel estimate. While at low SNR levels, additive white Gaussian noise remains the dominant factor, at high SNR levels, the aliasing artifacts arising from subsampling dominate. As the aliasing artifacts with PCS are substantially lower than RCS due to the optimized circular shifts, the disparity in performance between PCS and RCS is discernible at high SNR.
\section{Conclusions}\label{sec6_conclusion}
In this paper, we proposed a new comb structured codebook for sector level sweep with phased arrays. We also developed a convolutional CS-based method for channel estimation or beam alignment within the sector of interest. Our approach requires a lower training overhead than the classical exhaustive beam scanning methods and integrates well into the beam refinement protocol of the IEEE 802.11ad/ay standards. Furthermore, it overcomes the poor received SNR issue with standard CS by concentrating the transmitter's energy in the sector of interest. We also optimized the CS matrix by designing the circular shifts in convolutional CS to reduce the in-sector channel reconstruction error. Finally, we provided conditions that guarantee the successful recovery of the support of beamspace within the sector of interest. Our results indicate that the reconstruction error in the estimated in-sector channel is low when the antenna weights applied at the transmitter illuminate the sector almost uniformly. 
\section{Appendix}
\subsection{Proof of Lemma \ref{lemma1}}\label{app:1a}
Let $\mathbf{1}_{\rho_{\mathrm{e}}\times\rho_{\mathrm{a}}}$ denote a $\rho_{\mathrm{e}}\times\rho_{\mathrm{a}}$ all-ones matrix. We use $\Gamma(\omega_1,\omega_2)$ to denote the 2D discrete-space Fourier transform of $\mathbf{1}_{\rho_{\mathrm{e}}\times\rho_{\mathrm{a}}}$ where $\{\omega_1,~\omega_2\}\in[0,2\pi]$ \cite{jain1989fundamentals} are continuous. The inverse 2D-DFT of $\mathbf{1}_{\rho_{\mathrm{e}}\times\rho_{\mathrm{a}}}$ is $\mathbf{U}_{\rho_{\mathrm{e}}}^{*}\mathbf{1}_{\rho_{\mathrm{e}}\times\rho_{\mathrm{a}}}\mathbf{U}_{\rho_{\mathrm{a}}}^{*}=\sqrt{\rho_{\mathrm{e}}\rho_{\mathrm{a}}}\mathbf{E}_{\rho_{\mathrm{e}}\times\rho_{\mathrm{a}}}$ where the ${\rho_{\mathrm{e}}\times\rho_{\mathrm{a}}}$ matrix $\mathbf{E}_{\rho_{\mathrm{e}}\times\rho_{\mathrm{a}}}$ is zero at all entries except $E_{\rho_{\mathrm{e}}\times\rho_{\mathrm{a}}}(0,0)=1$. Hence, we have \cite[Sec. 5.5]{jain1989fundamentals}
\begin{equation}\label{eqn:sampledFT_1}
    \frac{1}{\sqrt{\rho_{\mathrm{e}}\rho_{\mathrm{a}}}}\Gamma^{*}(\frac{2\pi i}{\rho_{\mathrm{e}}},\frac{2\pi j}{\rho_{\mathrm{a}}}) = \sqrt{\rho_{\mathrm{e}}\rho_{\mathrm{a}}}E_{\rho_{\mathrm{e}}\times\rho_{\mathrm{a}}}(i,j),~~~~\forall (i,j)\in[\rho_{\mathrm{e}}]\times[\rho_{\mathrm{a}}].
\end{equation}
 Since $E_{\rho_{\mathrm{e}}\times\rho_{\mathrm{a}}}(i,j)=0~~\forall (i,j)\in[\rho_{\mathrm{e}}]\times[\rho_{\mathrm{a}}]\backslash(0,0)$, from \eqref{eqn:sampledFT_1}, we have
\begin{equation}\label{eqn:FT_zeros}
    \Gamma(\frac{2\pi i}{\rho_{\mathrm{e}}},\frac{2\pi j}{\rho_{\mathrm{a}}}) = 0,~~~~\forall (i,j)\in[\rho_{\mathrm{e}}]\times[\rho_{\mathrm{a}}]\backslash(0,0).
\end{equation}
Next, we observe that $\mathbf{N}_{\Omega}^{\mathrm{Nyq}}$ defined in Lemma \ref{lemma1} is an extended version of $\mathbf{1}_{\rho_{\mathrm{e}}\times\rho_{\mathrm{a}}}$ obtained by appending $\rho_{\mathrm{e}}(N_{\mathrm{e}}-1)$ zeros at the end of each row and then appending $\rho_{\mathrm{a}}(N_{\mathrm{a}}-1)$ zeros at the end of each column. Now, we use the fact that the 2D discrete-space Fourier transform $\Gamma(\omega_1,\omega_2)$ of the zero-padded signal remains the same, however, its 2D-DFT will be a densely sampled version of $\Gamma(\omega_1,\omega_2)$ \cite[Ch. 5]{jain1989fundamentals}. 
Hence, the PSF \eqref{eqn:psf} of $\mathbf{N}_{\Omega}^{\mathrm{Nyq}}$ which is a zero-padded version of $\mathbf{1}_{\rho_{\mathrm{e}}\times\rho_{\mathrm{a}}}$ can be written as
\begin{equation}\label{eqn:psf_sampled_FT}
    \mathrm{PSF}(i,j) = \frac{1}{M}\Gamma^{*}(\frac{2\pi i}{N},\frac{2\pi j}{N}),~~~~ (i,j)\in[N]\times[N].
\end{equation}
Finally, from  $N=\rho_{\mathrm{e}}N_{\mathrm{e}}$, $N=\rho_{\mathrm{a}}N_{\mathrm{a}}$ and \eqref{eqn:FT_zeros}, we observe that $\mathrm{PSF}(i,j)=0$ at the desired locations defined by $\mathcal{T}$ in \eqref{eqn:eff_coherence_2}. Hence, for $\mathbf{N}_{\Omega}^{\mathrm{Nyq}}$, the coherence defined in \eqref{eqn:eff_coherence_2} is equal to zero. Therefore, the sampling set $\Omega$ given in the lemma is optimal which concludes the proof.
\ifCLASSOPTIONcaptionsoff
  \newpage
\fi

\bibliography{References}
\bibliographystyle{ieeetr}

\end{document}